\documentclass[12pt]{iopart}

\expandafter\let\csname equation*\endcsname\relax
\expandafter\let\csname endequation*\endcsname\relax
\usepackage{amsmath}
\usepackage{amssymb}
\usepackage{amsthm}
\usepackage[normalem]{ulem}
\usepackage{hyperref}

\newtheorem{theorem}{Theorem}[section]
\newtheorem*{theorem*}{Theorem}
\newtheorem{lemma}[theorem]{Lemma}
\newtheorem{proposition}[theorem]{Proposition}
\newtheorem{corollary}[theorem]{Corollary}
\newtheorem{example}[theorem]{Example}

\theoremstyle{definition}
\newtheorem{definition}{Definition}[section]

\theoremstyle{remark}
\newtheorem*{remark*}{Remark}

\DeclareMathOperator*{\esssup}{ess\,sup}

\numberwithin{equation}{section}

\begin{document}

\title{On the generalized Langevin equation and the Mori projection operator technique}

\author{Christoph Widder\textsuperscript{1}, Johannes Zimmer\textsuperscript{2}, Tanja Schilling\textsuperscript{1}}
\address{\textsuperscript{1} Institute of Physics, Albert-Ludwigs-Universität Freiburg, Hermann-Herder-Straße 3, 79104 Freiburg im Breisgau, Germany}
\address{\textsuperscript{2} Department of Mathematics, Technische Universität München, Boltzmannstr. 3, 85748
Garching, Germany}
\eads{\mailto{christoph.widder@physik.uni-freiburg.de}, \mailto{jz@tum.de}, \mailto{tanja.schilling@physik.uni-freiburg.de}}

\vspace{10pt}
\begin{indented}
\item[]\today
\end{indented}

\begin{abstract}
In statistical physics, the Nakajima-Mori-Zwanzig projection operator formalism is used to derive an integro-differential equation for observables in a Hilbert space, the generalized Langevin equation (GLE). This technique relies on the splitting of the dynamics into a projected and an orthogonal part. However, the well-posedness of the abstract Cauchy problem for the orthogonal dynamics remains an open problem. Moreover, it is rarely discussed under which assumptions the Dyson identity, which is used to derive the GLE, holds. In this article, we address this issue for rank-one projections (Mori's projection). For the Mori projection, the orthogonal dynamics is obtained from the bounded perturbation theorem. The variation of constants formula for strongly continuous semigroups then yields the GLE and the second fluctuation dissipation theorem (2FDT). We show that the variation of constants can be replaced by a limiting process in order to give a general proof of the GLE and 2FDT that does not require the differentiability of the fluctuating forces. In addition, we offer an alternative approach that does not require the bounded perturbation theorem. Our starting point is the observation that the GLE and 2FDT uniquely determine the fluctuating forces as well as the memory kernel. Furthermore, the orbit maps for the orthogonal dynamics can be directly defined via solutions of linear Volterra equations. All desired properties of the orthogonal dynamics are then proven directly from this definition. In particular, the orthogonal dynamics is a strongly continuous semigroup generated by $\overline{\mathcal{QL}}\mathcal{Q}=\mathcal{QLQ}$, where $\mathcal{L}$ is the generator of the time evolution operator, and $\mathcal{P}=1-\mathcal{Q}$ is the Mori projection operator. If $\mathcal{L}$ is skew-adjoint, the orthogonal dynamics is a unitary group and the fluctuating forces are stationary. Our results apply to general autonomous dynamical systems whose time evolution is given by a strongly continuous semigroup. 
\end{abstract}

%
\vspace{2pc}

\clearpage
\section*{Addendum}

We have complemented this article with an addendum (\ref{sec:addendum}), which greatly simplifies the semigroup approach presented in section~\ref{ssec:semigroup_approach}. 

\section{Introduction}

In statistical physics, the method called ``projection operator formalism'', ``Mori-Zwanzig formalism'' or ``Nakajima-Zwanzig-formalism'' is a well-established method to reduce the dimension of a dynamical system. (For textbooks see ref.~\cite{grabert2006projection,zwanzig2001nonequilibrium,evans2007statistical}, for original texts see ref.~\cite{nakajima1958,zwanzig60,zwanzig1961,mori1965transport,kubo66,fukui1971,zwanzig72,zwanzig1973,kawasaki73,grabert77,grabert78}.) The basic idea of this method is to start out with the complete description of a microscopic system in terms of its Liouville-von Neumann equation and to produce the evolution equations of a set of observables of interest by means of a projection. The resulting evolution equations are called Generalized Langevin Equations (GLE). 

Since the 1960s there has been considerable research effort in analysing the properties of the GLE and in developing methods to infer the parameters of the GLE from experimental data (for didactic introductions to the topic see, e.g., \cite{snook2006langevin,schilling2022coarse,vrugt2020projection,chorin2000optimal,hijon2010mori}). In the physics literature, authors who write about the projection operator formalism usually base their arguments on an operator identity which they associate with the names Kawasaki, Dyson or Duhamel (depending on the authors' scientific background). This operator identity serves to handle the "orthogonal dynamics", i.e.~the dynamics of degrees of freedom orthogonal to the projection \cite{zwanzig60,mori1965transport,grabert77,holian85,evans84}. Surprisingly, there is hardly any literature in which the origin and range of validity of this operator identity or the properties of the orthogonal dynamics are mentioned. This raises the question for which systems the problem of solving the orthogonal dynamics is well-posed. To our knowledge, there is just one article that gives a proof of existence for the orthogonal dynamics \cite{givon2005}. We have the impression that large parts of the statistical physics community have accepted the Dyson identity as a justification for the GLE without questioning its mathematical background. This impression is supported by a statement in the work of Givon, Hald and Kupferman who motivate their existence proof by writing ``The validity of the Mori-Zwanzig formalism relies on the well-posedness of
the orthogonal dynamics equation, which has always been taken for granted.'' 
\cite{givon2005}. Also Zhu, Dominy, and Venturi remark ``When estimating the M[ori]Z[wanzig] memory integral, we need to deal with the semigroup $e^{tLQ}$. It turns
out to be extremely difficult to prove strong continuity of such a semigroup in general\dots'' \cite{zhu18}. (In their work, the term $e^{tLQ}$ is used to construct the orthogonal dynamics.) 

For Hamiltonian systems, Givon et al.~\cite{givon2005} have shown the existence of weak solutions to the orthogonal dynamics equation for conditional expectations. In particular, this includes Zwanzig's projection, which is of infinite rank. Their proof is closely related to Friedrich's extension for symmetric hyperbolic systems. The uniqueness of a solution, however, is not discussed. Therefore, it remains unclear if the fluctuating forces can be unambiguously defined by an orthogonal dynamics equation. For finite-rank projections, Givon et al. have proven the existence and uniqueness of a solution to the orthogonal dynamics equation by a reduction to linear Volterra equations. In contrast, Zhu et al.~\cite{zhu18} construct the orthogonal dynamics by means of the bounded perturbation theorem from semigroup theory. 

Here, we pick up these ideas to present two independent paths towards a general and rigorous formulation of the projection operator formalism for the Mori projection operator $\mathcal{P}:=(\cdot,z)(z,z)^{-1}z$, where $0\neq z \in H$ denotes the observable of interest for some complex Hilbert space $H$. Our results apply to arbitrary dynamical systems whose time evolution is given by a strongly continuous semigroup $\{\mathcal{U}(t)\}_{t\geq 0}$ on $H$. 

Firstly, we discuss a genuine semigroup approach in sec.~\ref{ssec:semigroup_approach}. The fluctuating forces can be defined by an orbit map of the orthogonal dynamics constructed via the bounded perturbation theorem. Afterwards, the GLE including the second fluctuation dissipation theorem (2FDT) is obtained from the variation of constants formula for strongly continuous semigroups. However, this requires the differentiability of the fluctuating forces, which turns out to be dispensable. In theorem \ref{theorem:gle_2}, we replace the variation of constants by a limiting process obtained from the exponential formula for strongly continuous semigroups. This shows that the GLE including the fluctuation dissipation theorem (2FDT) holds, provided that $z \in D(\mathcal{L})\cap D(\mathcal{L}^\dagger)$, where $\mathcal{L}$ denotes the generator of $\{\mathcal{U}(t)\}_{t\geq 0}$. Further, we show that the orthogonal dynamics is given by the strongly continuous semigroup $\{e^{\overline{\mathcal{QL}}\mathcal{Q}t}\}_{t\geq 0}$, where $\mathcal{Q}:=1-\mathcal{P}$ and the closure is denoted by an overbar, see lemmata \ref{lemma:strongly_continuous_semigroup}-\ref{lemma:strongly_continuous_semigroup_2} and corollary \ref{corollary:orthogonal_dynamics}. Consequently, the fluctuating forces can be defined via the unique mild solution of the abstract Cauchy problem associated to $\overline{\mathcal{QL}}\mathcal{Q}$ with initial value $\mathcal{QL}z$. 

Secondly, in section \ref{ssec:volterra}, we show that the same results can be obtained solely by means of linear Volterra equations. The derivation requires neither the bounded perturbation theorem nor the variation of constants or exponential formula. Instead, we start with the observation that the GLE and 2FDT uniquely determine the memory kernel as well as the fluctuating forces, see proposition \ref{proposition:gle}. This is an immediate consequence of the existence and uniqueness of solutions of linear Volterra equations with continuous coefficient functions. Hence, the GLE and 2FDT are valid by construction, except that the explicit formula for the fluctuating forces in terms of the orthogonal dynamics has yet to be established. Up to that point, the fluctuating force is defined via the unique solution of a linear Volterra equation. Thus, we first show that the fluctuating forces are orthogonal to the observable of interest $z$, see corollary \ref{corolloary:orthogonality}. Similar to the fluctuating forces, we can define all orbit maps $u(x,t)$ of the orthogonal dynamics by means of solutions of linear Volterra equations, where $x\in H$ denotes the initial value. In theorem \ref{theorem:mild_solution}, we show that these orbit maps are indeed mild solutions of the abstract Cauchy problem associated to $\overline{\mathcal{QL}}\mathcal{Q}$. In theorem \ref{theorem:strongly_continuous_semigroup}, we use Grönwall's inequality to show that $\{u(\cdot,t)\}_{t\geq 0}$ constitutes a strongly continuous semigroup of bounded linear operators. Moreover, the generator of $\{u(\cdot,t)\}_{t\geq 0}$ is identified with $\overline{\mathcal{QL}}\mathcal{Q}=\mathcal{QLQ}$, i.e. $u(x,t)=e^{\overline{\mathcal{QL}}\mathcal{Q}t}x$. Again, we find that the fluctuating forces are given by the unique mild solution of the abstract Cauchy problem associated to $\overline{\mathcal{QL}}\mathcal{Q}$ with initial value $\mathcal{QL}z$, see corollary \ref{corollary:fluctuating_forces}. This restores the missing formula for the fluctuating forces in terms of the orthogonal dynamics. The final form of the GLE and 2FDT, theorem \ref{theorem:gle}, is now a direct consequence of proposition \ref{proposition:gle} and corollary \ref{corollary:fluctuating_forces}. Additionally, in sec.~\ref{ssec:unitary_time_evolution}, we prove that the orthogonal dynamics is a unitary group if the time evolution of the dynamical system is a strongly continuous unitary group. In this case, the fluctuating forces are stationary, and the generator of the orthogonal dynamics is given by $\overline{\mathcal{QL}}\mathcal{Q}=\overline{\mathcal{QLQ}}$. 

In sec.~\ref{sec:classical_statistical_mechanics}, we show how these results apply to classical statistical mechanics. The non-autonomous case (e.g., systems under time-dependent external driving or active matter) is treated in sec.~\ref{sec:non_autonomous_systems}. 
The same way as for autonomous systems, we show that the non-stationary versions of the GLE and 2FDT \cite{mcphie2001,kawai2011,meyer2017,meyer2019,vrugt2019mori,widder22} uniquely determine the memory kernel and the fluctuating forces such that the orthogonality between the fluctuating forces and the observable of interest holds. Hence, the non-stationary GLE and 2FDT hold by construction, except that the fluctuating forces are defined via solutions of linear Volterra equations. For non-autonomous systems, however, we cannot provide a rigorous definition of the fluctuating forces by means of a well-posed orthogonal dynamics equation. Finally, a discussion and conclusion is given in sec.~\ref{sec:conclusion}.

\section{Semigroup approach}\label{ssec:semigroup_approach}

Throughout this article, let $H$ be a complex Hilbert space, and let $(.,.)$ denote the scalar product with complex conjugation in its second argument. We write $c^\ast$ for the conjugate of a complex number $c$. Let $\{\mathcal{U}(t)\}_{t\geq 0}$ be a strongly continuous semigroup on $H$ such that the time evolution for some state vector $x\in H$ is given by the orbit map $t\to \mathcal{U}(t)x$. 
The generator is defined by
\begin{align*}
    \mathcal{L}x &:= \lim_{h\searrow 0} \frac{1}{h}[\mathcal{U}(h)x-x] \intertext{on the domain}
    D(\mathcal{L}) &:= \{ x\in H \colon \mathcal{L}x \in H \} \, .
\end{align*}
 This definition follows the convention from semigroup theory as used, e.g., in \cite[p.~49]{engel}. Within the physics literature (see, e.g., \cite[p.~208]{hall2013}), the generator is typically defined by $\mathcal{L}x:=\lim_{h\searrow 0} \frac{1}{ih}[\mathcal{U}(h)x-x]$. This convention is retrieved by the replacement $\mathcal{L} \mapsto i\mathcal{L}$. 

 Further, let $\mathcal{P}$ be an orthogonal projection on $H$. Its complementary projection is denoted by $\mathcal{Q}:=1-\mathcal{P}$. For simplicity, we restrict ourselves to rank-one projections
\begin{align}
    \mathcal{P} &:= (\cdot,z)(z,z)^{-1}z \, , \label{def:mori_projector}
\end{align}
where $0\neq z \in H$ denotes some fixed observable of interest. The operator $\mathcal{P}$ is also called Mori projection operator. For simplicity, we omit the dependence of $\mathcal P$ on $z$ in the notation.

For $z\in D(\mathcal{L})$, the semigroup generated by $\mathcal{LQ}=\mathcal{L}-\mathcal{LP}$ can be considered a bounded perturbation of the semigroup $\{\mathcal{U}(t)\}_{t\geq 0}$, see Zhu et al. \cite[Appendix A]{zhu18}. Hence, the bounded perturbation theorem \cite[p.~158, Theorem~1.3]{engel} yields the strongly continuous semigroup $\{e^{\mathcal{LQ}t}\}_{t\geq 0}$. This remains true for finite-rank projections. In addition, the bounded perturbation theorem entails a growth bound for the perturbed semigroup, which is useful to derive estimates on the memory integral \cite{zhu18}.   

\begin{remark*}
    Mind that $\{e^{\mathcal{LQ}t}\}_{t\geq 0}$ is not given by the exponential series. Here, the notation $\{e^{\mathcal{LQ}t}\}_{t\geq 0}$ means that $\mathcal{LQ}$ is the generator of a strongly continuous semigroup. Furthermore, the order of the operators matters. E.g., since $\mathcal{PL}$ is generally unbounded, the bounded perturbation theorem does not imply that $\mathcal{QL}$ generates a strongly continuous semigroup. 
\end{remark*}

The GLE is typically derived using the Dyson identity \cite{grabert2006projection,zwanzig2001nonequilibrium,snook2006langevin}, which resembles a variation of constants \cite{ball}. Let us for the moment ignore possible domain issues. Let us define the fluctuating forces according to $\eta_t:=\mathcal{Q}e^{\mathcal{LQ}t}\mathcal{QL}z$. The derivative of $\eta_t$ takes the form
\begin{align*}
    \frac{d}{dt}\eta_t &= \mathcal{QLQ}e^{\mathcal{LQ}t}\mathcal{QL}z = 
    (\mathcal{L} - \mathcal{PL})\mathcal{Q}e^{\mathcal{LQ}t}\mathcal{QL}z=\mathcal{L}\eta_t-\mathcal{PL}\eta_t \, .
\end{align*}
Applying the variation of constants formula \cite{ball} yields
\begin{align*}
    \eta_t &= \mathcal{U}(t)\mathcal{QL}z - \int^t_0 \mathcal{U}(t-s)\mathcal{PL}\eta_s \,  ds \, .
\end{align*}
Insertion into the time derivative of $\mathcal{U}(t)z$ yields the GLE, 
\begin{align*}
    \frac{d}{dt}\mathcal{U}(t)z &= \mathcal{U}(t)\mathcal{PL}z + \mathcal{U}(t)\mathcal{QL}z \\
    &= \mathcal{U}(t)\mathcal{PL}z + \eta_t +  \int^t_0\mathcal{U}(t-s)\mathcal{PL}\eta_s \,  ds \\
    &= \mathcal{U}(t)\mathcal{PL}z + \eta_t +  \int^t_0(\mathcal{L}\eta_s,z)(z,z)^{-1}\mathcal{U}(t-s)z \,  ds \\
    &= \mathcal{U}(t)\mathcal{PL}z + \eta_t + \int^t_0 K(t-s)\mathcal{U}(s)z \, ds \, ,
\end{align*}
where we introduced the memory kernel $K(t)$ according to
\begin{align}
    K(t) &:= (\mathcal{L}\eta_t,z)(z,z)^{-1}  = (\eta_t,\mathcal{L}^\dagger z)(z,z)^{-1} =  (\eta_t,\mathcal{Q}\mathcal{L}^\dagger z)(z,z)^{-1}   \, ; \label{equ:2fdt}
\end{align}
the last equality holds since the range of $\eta$ is in ${\mathcal Q} H$.
If $\mathcal{L}$ is skew-adjoint, we retrieve the second fluctuation dissipation theorem (2FDT), i.e., $K(t) = -(\eta_t,\eta_0)(z,z)^{-1}$ \cite{kubo66}. Hence, eq.~(\ref{equ:2fdt}) is a direct generalization of the 2FDT for arbitrary strongly continuous semigroups $\{\mathcal{U}(t)\}_{t\geq 0}$. 
We can now introduce the orthogonal dynamics according to 
\begin{align}
    \mathcal{G}(t) &:= \mathcal{P} +  \mathcal{Q}e^{\mathcal{LQ}t} \, , \label{equ:orthogonal_dynamics}
\end{align}
where we have included the operator $\mathcal{P}$ for convenience. The inclusion of $\mathcal{P}$ assures that the orthogonal dynamics leaves the parallel part of the initial value invariant. In particular, we have $\mathcal{G}(0)=1$. This allows us to treat the orthogonal dynamics as a semigroup on the full Hilbert space $H$. 

For this derivation of the GLE and 2FDT to be valid, the differentiability of the fluctuating force $\eta_t$ is required. In the following, we give a more general proof using a limiting process obtained from the exponential formula for strongly continuous semigroups \cite{pazy}. 

First, let us show that the orthogonal dynamics $\{\mathcal{G}(t)\}_{t\geq 0}$ from eq.~\eqref{equ:orthogonal_dynamics} is again a strongly continuous semigroup. 

\begin{lemma}\label{lemma:strongly_continuous_semigroup}
     Let $z\in D(\mathcal{L})$ with $z\neq0$. Then, $\{\mathcal{G}(t) := \mathcal{P}+\mathcal{Q}e^{\mathcal{LQ}t}\}_{t\geq 0}$ is a strongly continuous semigroup.
\end{lemma}
\begin{proof}
    Since $z\in D(\mathcal{L})$ the operator $\mathcal{LP}$ is bounded. Thus, it follows by the bounded perturbation theorem \cite[p.~158, Theorem~1.3]{engel} that $\mathcal{LQ}=\mathcal{L}-\mathcal{LP}$ generates a strongly continuous semigroup. Hence, $\{\mathcal{P}+\mathcal{Q}e^{\mathcal{LQ}t}\}_{t\geq 0}$ is a one-parameter family of bounded linear operators. 
    
    It is left to show that $\{\mathcal{P}+\mathcal{Q}e^{\mathcal{LQ}t}\}_{t\geq 0}$ is a semigroup. For all $x\in \mathcal{P}H$, we have $x\in D(\mathcal{LQ})$ and $\mathcal{LQ}x=0$. This implies that the derivative of $t\to e^{\mathcal{LQ}t}x$ vanishes for all $x\in\mathcal{P}H$ and $t\geq 0$. Consequently, we have $e^{\mathcal{LQ}t}x=x$ for all $x\in\mathcal{P}H$ and $t\geq 0$. This implies for all $x\in H$ and $t,s \geq 0$
    \begin{align*}
        \left(\mathcal{P}+\mathcal{Q}e^{\mathcal{LQ}t}\right)\left(\mathcal{P}+\mathcal{Q}e^{\mathcal{LQ}s}\right)x &=\mathcal{P}x+\mathcal{Q}e^{\mathcal{LQ}t}\mathcal{Q}e^{\mathcal{LQ}s}x +\underbrace{\mathcal{Q}e^{\mathcal{LQ}t}\mathcal{P}x}_{=0} \\
        &= \mathcal{P}x +\mathcal{Q}e^{\mathcal{LQ}t}e^{\mathcal{LQ}s}x -\underbrace{\mathcal{Q}e^{\mathcal{LQ}t}\mathcal{P}e^{\mathcal{LQ}s}x}_{=0} \\
        &= \mathcal{P}x +\mathcal{Q}e^{\mathcal{LQ}(t+s)}x \, .
    \end{align*}
    This completes the proof. 
\end{proof}

An alternative representation for the orthogonal dynamics can be obtained using adjoint semigroups. In addition, we find that the generator of the orthogonal dynamics is given by $\overline{\mathcal{QL}}\mathcal{Q}$. Before we proceed, we require the following proposition.
\begin{proposition}\label{proposition:props_of_generator}
    Let $z\in D(\mathcal{L}^\dagger)$ with $z\neq0$. Then the operators $\mathcal{L}^\dagger$,  $[\mathcal{L}^\dagger \mathcal{Q}]^\dagger=\overline{\mathcal{QL}}$, $\mathcal{L}^\dagger\mathcal{Q}=\overline{\mathcal{QL}}^\dagger$ and $\overline{\mathcal{QL}}\mathcal{Q}=(\mathcal{Q}\mathcal{L}^\dagger\mathcal{Q})^\dagger$ are densely defined and closed. 
\end{proposition}
\begin{proof}
Since $H$ is reflexive, the adjoint semigroup $\{\mathcal{U}(t)^\dagger\}_{t\geq 0}$ is strongly continuous \cite[p.~6, Corollary 1.3.2]{vanNeerven1992} with generator $\mathcal{L}^\dagger$ \cite[p.~6, Theorem~1.3.3]{vanNeerven1992}, see also \cite[p.~43-44, p.~61-64]{engel}. Since $z\in D(\mathcal{L}^\dagger)$, the operator $\mathcal{L}^\dagger \mathcal{P}$ is bounded. 
Since $\mathcal{L}^\dagger$ is densely defined and closed \cite[p.~51, Theorem~1.4]{engel}, this implies that $\mathcal{L}^\dagger\mathcal{Q}= \mathcal{L}^\dagger- \mathcal{L}^\dagger \mathcal{P}$ is also densely defined and closed.  Since $\mathcal{Q}$ is bounded and $\mathcal{L}$ is densely defined, this implies $[\mathcal{L}^\dagger\mathcal{Q}]^\dagger = {[\mathcal{QL}]^\dagger}^\dagger$ \cite[p.~330, Theorem~13.2]{rudin1974}. Since $\mathcal{QL}$ and $[\mathcal{QL}]^\dagger=\mathcal{L}^\dagger\mathcal{Q}$ are densely defined, $\mathcal{QL}$ is closable with closure $\overline{\mathcal{QL}}={[\mathcal{QL}]^\dagger}^\dagger$ \cite[p.~252, Theorem VIII.1]{reed1980}. Hence, $[\mathcal{L}^\dagger\mathcal{Q}]^\dagger = \overline{\mathcal{QL}}$. Since $\mathcal{L}^\dagger\mathcal{Q}$ is densely defined and closed, this implies $ \mathcal{L}^\dagger \mathcal{Q} = \overline{\mathcal{QL}}^\dagger$. Furthermore, we have 
\begin{align}
\label{equ:L_dagger_Q_dagger}(\mathcal{Q}\mathcal{L}^\dagger\mathcal{Q})^\dagger &=  [\mathcal{L}^\dagger\mathcal{Q}]^\dagger \mathcal{Q} = \overline{\mathcal{QL}}\mathcal{Q} \, ,
\end{align}
where we used \cite[p.~330, Theorem~13.2]{rudin1974}.
\end{proof}

\begin{lemma}\label{lemma:strongly_continuous_semigroup_2}
    Let $z\in D(\mathcal{L}^\dagger)$ with $z\neq0$. Then, $\{\mathcal{G}(t) := \mathcal{P}+e^{\overline{\mathcal{QL}}t}\mathcal{Q}\}_{t\geq 0}$ is a strongly continuous semigroup with generator $\overline{\mathcal{QL}}\mathcal{Q}$.
\end{lemma}
\begin{proof}
     Since $z\in D(\mathcal{L}^\dagger)$, it follows by lemma \ref{lemma:strongly_continuous_semigroup} that $\{\mathcal{P}+\mathcal{Q}e^{\mathcal{L^\dagger Q}t}\}_{t\geq 0}$ is a strongly continuous semigroup. Again, the operator $[\mathcal{L}^\dagger \mathcal{Q}]^\dagger$ generates the adjoint semigroup to $\{e^{\mathcal{L}^\dagger \mathcal{Q}t}\}_{t\geq 0}$. By proposition \ref{proposition:props_of_generator}, we have $[\mathcal{L}^\dagger \mathcal{Q}]^\dagger=\overline{\mathcal{QL}}$. Furthermore, the adjoint semigroup to $\{\mathcal{P}+\mathcal{Q}e^{\mathcal{L^\dagger Q}t}\}_{t\geq 0}$ is a strongly continuous semigroup. Thus, 
    \begin{align*}
        \left(\mathcal{P}+\mathcal{Q}e^{\mathcal{L^\dagger Q}t}\right)^\dagger &=\mathcal{P} + \left(e^{\mathcal{L^\dagger Q}t}\right)^\dagger\mathcal{Q} \\
        &=  \mathcal{P} + e^{[\mathcal{L^\dagger Q}]^\dagger t}\mathcal{Q} \\
        &= \mathcal{P} + e^{\overline{\mathcal{QL}}t}\mathcal{Q} \\
        &= \mathcal{G}(t)\, .
    \end{align*}
    This shows that $\{\mathcal{G}(t)\}_{t\geq 0}$ is a strongly continuous semigroup. 

    It is left to show that $\overline{\mathcal{QL}}\mathcal{Q}$ is the generator. Note that the orbit map $t\to\mathcal{G}(t)x$ is right differentiable at $t=0$, if and only if $x\in D(\overline{\mathcal{QL}}\mathcal{Q}):= \{x\in H \,:\, \mathcal{Q}x\in D(\overline{\mathcal{QL}})\}$. In this case, the right derivative at $t=0$ is given by $\overline{\mathcal{QL}}\mathcal{Q}x$. This shows that $\overline{\mathcal{QL}}\mathcal{Q}$ is the generator of $\{\mathcal{G}(t)\}_{t\geq 0}$.
\end{proof}
For $z\in D(\mathcal{L})\cap D(\mathcal{L}^\dagger)$, it is straightforward to prove the equality of the semigroups $\{\mathcal{G}(t)\}_{t\geq 0}$ from lemmata \ref{lemma:strongly_continuous_semigroup} and \ref{lemma:strongly_continuous_semigroup_2}.

\begin{corollary}\label{corollary:orthogonal_dynamics}
    Let $z\in D(\mathcal{L}) \cap D(\mathcal{L}^\dagger)$  with $z\neq0$. Then, for all $t\geq 0$
    \begin{align}
        \mathcal{G}(t) &= \mathcal{P} +  \mathcal{Q}e^{\mathcal{LQ}t} \, ,
    \end{align}
    where $\{\mathcal{G}(t)\}_{t\geq 0}$ denotes the strongly continuous semigroup generated by $\overline{\mathcal{QL}}\mathcal{Q}$, and $\{e^{\mathcal{LQ}t}\}_{t\geq 0}$ denotes the strongly continuous semigroup generated by $\mathcal{LQ}$. 
\end{corollary}
\begin{proof}
    Let $x\in D(\mathcal{LQ})$ arbitrary. Since $\mathcal{Q}$ is bounded, we have
    \begin{align*}
       \frac{d}{dt}[ \mathcal{P}x +  \mathcal{Q}e^{\mathcal{LQ}t}x ] &= \mathcal{QLQ}[ \mathcal{P}x + \mathcal{Q}e^{\mathcal{LQ}t}x] \, .
    \end{align*}
    This implies that $t\to \mathcal{P}x +  \mathcal{Q}e^{\mathcal{LQ}t}x$ is a mild solution of the abstract Cauchy problem associated to $\overline{\mathcal{QL}}\mathcal{Q}$ \cite[p.~50, p.~145-146]{engel}. However, since $\overline{\mathcal{QL}}\mathcal{Q}$ generates a strongly continuous semigroup, the mild solution is unique \cite[p.~146, Proposition 6.4]{engel}. Hence, we have
    \begin{align}
        \mathcal{G}(t)x &= \mathcal{P}x +  \mathcal{Q}e^{\mathcal{LQ}t}x \, ,
    \end{align}
    for all $x\in D(\mathcal{LQ})$ and $t\geq 0$. Since $\mathcal{LQ}$ generates a strongly continuous semigroup, $\mathcal{LQ}$ is densely defined \cite[p.~51, Theorem~1.4]{engel}. This implies 
    \begin{align*}
        \|\mathcal{G}(t) -\mathcal{P}-\mathcal{Q}e^{\mathcal{LQ}t}\| &= 0 \, ,
    \end{align*}
    for all $t\geq 0$. This concludes the proof. 
\end{proof}

This implies that the fluctuating force $\eta_t:=\mathcal{G}(t)\mathcal{QL}z$ is the unique mild solution of the abstract Cauchy problem associated to $\overline{\mathcal{QL}}\mathcal{Q}$ with initial value $\mathcal{QL}z$, cf. sec.~\ref{ssec:volterra}. 

\begin{theorem}[GLE and 2FDT]\label{theorem:gle_2}
    Let $z\in D(\mathcal{L})\cap D(\mathcal{L}^\dagger)$ with $z\neq0$. Then, the orbit map $\mathbb{R}_+ \ni t \mapsto \mathcal{U}(t)z \in D(\mathcal{L})$ solves the integro-differential equation
    \begin{align}
        \frac{d}{dt}\mathcal{U}(t)z &= \mathcal{U}(t)\mathcal{P}\mathcal{L}z + \mathcal{G}(t)\mathcal{Q}\mathcal{L}z+\int^t_0   K(t-s)\mathcal{U}(s)z \, ds \, , \\
        K(t)&:= (\mathcal{G}(t)\mathcal{Q}\mathcal{L}z,\mathcal{Q}\mathcal{L}^\dagger z)(z,z)^{-1} \, , \label{def:gle_2:memory_kernel} 
    \end{align}
    where $\{\mathcal{G}(t)\}_{t\geq 0}$ denotes the strongly continuous semigroup generated by $\overline{\mathcal{QL}}\mathcal{Q}$. 
\end{theorem}
\begin{proof}
    In order to avoid the requirement of differentiability of the fluctuating forces, it is useful to consider the approximants from the exponential formula for strongly continuous semigroups \cite[p.~33, Theorem~8.3]{pazy}
    \begin{align*}
        \mathcal{U}_n(t) &:= \left[\frac{n}{t}\mathcal{R}\left(\frac{n}{t};\mathcal{L}\right)\right]^n \, ,
    \end{align*}
    where $\mathcal{R}\left(\lambda;\mathcal{L}\right)=(\lambda-\mathcal{L})^{-1}$ denotes the resolvent. For $\lambda \in \rho(\mathcal{L})$, where $\rho(\mathcal{L})$ is the resolvent set, the resolvent $\mathcal{R}\left(\lambda;\mathcal{L}\right)\colon H \to D(\mathcal{L})$ is a bijection. Thus, the resolvent $\mathcal{R}\left(\lambda;\mathcal{L}\right)$ is bounded for all $\lambda\in\rho(\mathcal{L})$, and $D(\mathcal{L})$ is an invariant subspace for the map $\mathcal{U}_n(t)$ for all $\frac{n}{t}\in \rho(\mathcal{L})$. Furthermore, we have $\mathcal{R}'\left(\lambda;\mathcal{L}\right)=-\mathcal{R}\left(\lambda;\mathcal{L}\right)^2$  \cite[p.~56, Corollary 1.11]{engel} and $\lambda\mathcal{R}\left(\lambda;\mathcal{L}\right)-1=\mathcal{L}\mathcal{R}\left(\lambda;\mathcal{L}\right)$. Since $\mathcal{R}'\left(\lambda;\mathcal{L}\right)$ commutes with $\mathcal{R}\left(\lambda;\mathcal{L}\right)$, a straightforward calculation shows that
    \begin{align*}
        \frac{d}{dt}\mathcal{U}_n(t) &= \left[\frac{n}{t}\mathcal{R}\left(\frac{n}{t};\mathcal{L}\right)\right]^n\left[\frac{n^2}{t^2}\mathcal{R}\left(\frac{n}{t};\mathcal{L}\right)-\frac{n}{t}\right] \\
        &= \mathcal{U}_n(t)\mathcal{L}_n(t) \, , \\
        \mathcal{L}_n(t) &:= \mathcal{L}\left[\frac{n}{t}\mathcal{R}\left(\frac{n}{t};\mathcal{L}\right)\right] \, .
    \end{align*}

    The idea is to derive a suitable equation of motion for the approximate orbit maps $t\to\mathcal{U}_n(t)z$ such that the final result is obtained from the limit $n\to\infty$. Computing the derivative of the approximate orbit map $\mathcal{U}_n(t)z$ yields
    \begin{align}
        \frac{d}{dt} \mathcal{U}_n(t)z &= \mathcal{U}_n(t)\mathcal{L}_n(t)z = \mathcal{U}_n(t)\mathcal{P}\mathcal{L}_n(t)z+\mathcal{U}_n(t)\mathcal{Q}\mathcal{L}_n(t)z \, .  \label{equ:proof_gle_derivative}
    \end{align}
    Since $z\in D(\mathcal{L})$, the subspace $D(\mathcal{L})$ is invariant under the map $\mathcal{Q}$. Since $\mathcal{L}_n(t)=\left[\frac{n^2}{t^2}\mathcal{R}\left(\frac{n}{t};\mathcal{L}\right)-\frac{n}{t}\right]$, the subspace $D(\mathcal{L})$ is also invariant under the maps $\mathcal{L}_n(t)$. Thus, $\mathcal{Q}\mathcal{L}_n(t)z \in D(\mathcal{L})\cap \mathcal{Q}H$. The generator of $\{\mathcal{G}(t)\}_{t\geq0}$ is given by
    \begin{align*}
        \mathcal{A} &:= \overline{\mathcal{QL}}\mathcal{Q}  \, .
    \end{align*}
    Thus, we conclude $\mathcal{Q}\mathcal{L}_n(t)z \in D(\mathcal{A})$, and we may compute
    \begin{align}
        \frac{d}{ds}\mathcal{U}_n(s)\mathcal{G}(t-s)\mathcal{Q}\mathcal{L}_n(t)z &= \mathcal{U}_n(s)(\mathcal{L}_n(s)-\mathcal{A})\mathcal{G}(t-s)\mathcal{Q}\mathcal{L}_n(t)z \, . \label{equ:proof_gle_duhamel_0}
    \end{align}
    The maps $s\to \mathcal{G}(t-s)\mathcal{Q}\mathcal{L}_n(t)z$ and $s\to \mathcal{A}\mathcal{G}(t-s)\mathcal{Q}\mathcal{L}_n(t)z$ are continuous. Furthermore, the maps  $s\to \mathcal{U}_n(s)$ and $s\to \mathcal{L}_n(s)$ are continuous w.r.t. the operator norm.
    This shows that $[0,t]\ni s\to\mathcal{U}_n(s)\mathcal{G}(t-s)\mathcal{Q}\mathcal{L}_n(t)z$ is \textit{continuously} differentiable for sufficiently large $n$. Thus, integrating eq.~(\ref{equ:proof_gle_duhamel_0}) over $[0,t]$ yields
    \begin{align}     \mathcal{U}_n(t)\mathcal{Q}\mathcal{L}_n(t)z-\mathcal{G}(t)\mathcal{Q}\mathcal{L}_n(t)z &= \int^t_0   \mathcal{U}_n(s)(\mathcal{L}_n(s)-\mathcal{A})\mathcal{G}(t-s)\mathcal{Q}\mathcal{L}_n(t)z \, ds\, , \label{equ:proof_gle_duhamel}
    \end{align}
    where we use the Bochner integral \cite{Aliprantis2006, hille, evans}. Inserting eq.~(\ref{equ:proof_gle_duhamel}) into eq.~(\ref{equ:proof_gle_derivative}) yields
    \begin{align}
        \frac{d}{dt} \mathcal{U}_n(t)z &= \mathcal{U}_n(t)\mathcal{P}\mathcal{L}_n(t)z+\mathcal{G}(t)\mathcal{Q}\mathcal{L}_n(t)z+\int^t_0  \mathcal{U}_n(s)(\mathcal{L}_n(s)-\mathcal{A})\mathcal{G}(t-s)\mathcal{Q}\mathcal{L}_n(t)z \, ds\, .\label{equ:proof_gle_approximation}
    \end{align}
    The plan is to take the limit $n\to\infty$. 
    
    Let $T>0$ be arbitrary. According to the exponential formula \cite[p.~33, Theorem~8.3]{pazy}, for all $x\in H$, the limit
    \begin{align}
       \lim_{n\to\infty} \mathcal{U}_n(t) x &= \mathcal{U}(t)x  \label{equ:proof_gle_4}
    \end{align}
    converges uniformly on $[0,T]$. We have to show that $\|\mathcal{U}_n(t)\|$ is uniformly bounded for all $n \in \mathbb{N}$ and $t\in [0,T]$. Since $\mathcal{L}$ generates a strongly continuous semigroup, we have the estimates \cite[p.~77, Theorem~3.8]{engel}
    \begin{align*}
        \| (\lambda - \omega)^n\mathcal{R}(\lambda;\mathcal{L})^n \| & \leq M \, ,
    \end{align*}
    for all $\lambda>\omega$ and $\lambda \in \rho(\mathcal{L})$, where $\omega\in\mathbb{R}$ and $M\geq1$ are chosen such that $\|\mathcal{U}(t)\|\leq Me^{\omega t}$ for all $t\geq 0$. Thus, we have
    \begin{align*}
        \| \mathcal{U}_n(t)\| &= \left\| \left[\frac{n}{t}\mathcal{R}\left(\frac{n}{t};\mathcal{L}\right)\right]^n \right\| \leq M \left(\frac{1}{1-\frac{\omega t}{n}}\right)^n \, , 
    \end{align*}
    for all $\frac{n}{t} \in \rho(\mathcal{L})$ with $\frac{n}{t} > \omega$. 
    Since $\left(\frac{1}{1-\frac{\omega t}{n}}\right)^n \to e^{\omega t}$ monotonically, it follows that there exists $N\in\mathbb{N}$ such that $\|\mathcal{U}_n(t)\|$ is uniformly bounded for $n>N$ and $t\in [0,T]$. Hence, since $z\in D(\mathcal{L})$, for all $\varepsilon>0$, there exists $N\in\mathbb{N}$, such that for all $n> N$ and $t\in [0,T]$
    \begin{align*}
       \left\| \frac{d}{dt} \mathcal{U}_n(t)z - \frac{d}{dt}\mathcal{U}(t)z \right\| &=\left\| \mathcal{U}_n(t)\mathcal{L}_n(t) z - \mathcal{U}(t)\mathcal{L} z \right\| \\
       &\leq \| \mathcal{U}_n(t)\|\|\left(\mathcal{L}_n(t)- \mathcal{L} \right)z\|+ \| \left(\mathcal{U}_n(t)-\mathcal{U}(t)\right)\mathcal{L}z \| < \varepsilon   \, ,
    \end{align*}
    where we used the fact that for $z\in D(\mathcal{L})$, the limit
    \begin{align*}
       \lim_{n\to \infty}\mathcal{L}_n(t)z &= \lim_{n\to \infty}\left[\frac{n}{t}\mathcal{R}\left(\frac{n}{t};\mathcal{L}\right)\right]\mathcal{L}z = \mathcal{L}z
    \end{align*}
    is uniform on bounded intervals \cite[p.~9, Lemma~3.2]{pazy}. In other words, the limit 
    \begin{align}
       \lim_{n\to\infty} \frac{d}{dt} \mathcal{U}_n(t)z &= \frac{d}{dt} \mathcal{U}(t)z \label{equ:proof_gle_limit_of_derivative} 
    \end{align}
    converges uniformly on bounded intervals. Similarly, we show that the limits
    \begin{align}
        \lim_{n\to\infty} \mathcal{U}_n(t)\mathcal{P}\mathcal{L}_n(t)z &= \mathcal{U}(t)\mathcal{P}\mathcal{L}z \, , \label{equ:proof_gle_limit_of_drift}\\
    \lim_{n\to\infty}\mathcal{G}(t)\mathcal{Q}\mathcal{L}_n(t)z &= \mathcal{G}(t)\mathcal{Q}\mathcal{L}z \label{equ:proof_gle_limit_of_ff}
    \end{align}
    converge uniformly on bounded intervals. Taking the limit $n\to \infty$ of eq.~(\ref{equ:proof_gle_approximation}) and inserting eqs.~(\ref{equ:proof_gle_limit_of_derivative})-(\ref{equ:proof_gle_limit_of_ff}) yields
    \begin{align}
        \frac{d}{dt}\mathcal{U}(t)z &=\mathcal{U}(t)\mathcal{P}\mathcal{L}z+ \mathcal{G}(t)\mathcal{Q}\mathcal{L}z+\lim_{n\to\infty}\int^t_0   \mathcal{U}_n(s)(\mathcal{L}_n(s)-\mathcal{A})\mathcal{G}(t-s)\mathcal{Q}\mathcal{L}_n(t)z \, ds \, . \label{equ:proof_gle_approximation_2}
    \end{align}

    Now let $x \in D(\mathcal{L}^\dagger)$ arbitrary. Recall that $\{\mathcal{U}(s)^\dagger\}_{s\geq 0}$ is a strongly continuous semigroup with generator $\mathcal{L}^\dagger$ \cite{vanNeerven1992,engel}. Hence, using the exponential formula \cite[p.~33, Theorem~8.3]{pazy} and the fact that $\mathcal{R}(n/t;\mathcal{L})^\dagger=\mathcal{R}(n/t;\mathcal{L}^\dagger)$ \cite[p.~38, Lemma 10.2]{pazy}, we find that the limits
    \begin{align}
            \lim_{n\to\infty} \mathcal{U}_n(s)^\dagger x & =  \mathcal{U}(s)^\dagger x  \label{equ:proof_gle_0} \, , \\
        \lim_{n\to\infty} \mathcal{L}_n(s)^\dagger \mathcal{U}_n(s)^\dagger x & = \mathcal{L}^\dagger \mathcal{U}(s)^\dagger x  \label{equ:proof_gle_1}
    \end{align}
    converge uniformly on bounded intervals. By proposition \ref{proposition:props_of_generator} we have $\mathcal{A}=(\mathcal{Q}\mathcal{L}^\dagger\mathcal{Q})^\dagger$ and $\mathcal{A}^\dagger=\overline{\mathcal{Q}\mathcal{L}^\dagger \mathcal{Q}}$. Since $z\in D(\mathcal{L}^\dagger)$, the subspace $D(\mathcal{L}^\dagger)$ is invariant under the map $\mathcal{Q}$. Since $\mathcal{U}_n(s)^\dagger=\left[\frac{n}{s}\mathcal{R}\left(\frac{n}{s};\mathcal{L}^\dagger\right)\right]^n$, the subspace $D(\mathcal{L}^\dagger)$ is invariant under the maps $\mathcal{U}_n(s)^\dagger$.  Hence, $\mathcal{Q}\mathcal{U}_n(s)^\dagger x \in D(\mathcal{L}^\dagger)$ and $\mathcal{U}_n(s)^\dagger x \in D(\mathcal{A}^\dagger)$. Thus, the limit
    \begin{align}
        \lim_{n\to\infty} \mathcal{A}^\dagger \mathcal{U}_n(s)^\dagger x 
        &=\lim_{n\to\infty} \mathcal{Q}\mathcal{L}^\dagger\mathcal{Q} \mathcal{U}_n(s)^\dagger x\nonumber\\
        &=\lim_{n\to\infty} \mathcal{Q}\mathcal{L}^\dagger\mathcal{U}_n(s)^\dagger x - \lim_{n\to\infty} \mathcal{Q}\mathcal{L}^\dagger \mathcal{P}\mathcal{U}_n(s)^\dagger x \nonumber\\
        &= \lim_{n\to\infty} \mathcal{Q}\mathcal{U}_n(s)^\dagger \mathcal{L}^\dagger x - \lim_{n\to\infty} (\mathcal{U}_n(s)^\dagger x,z)(z,z)^{-1} \mathcal{Q}\mathcal{L}^\dagger z  \nonumber \\
        &= \mathcal{Q}\mathcal{U}(s)^\dagger \mathcal{L}^\dagger x - (\mathcal{U}(s)^\dagger x,z)(z,z)^{-1} \mathcal{Q}\mathcal{L}^\dagger z \, \label{equ:proof_gle_2}
    \end{align}
    is uniform on bounded intervals. Since $\mathcal{A}^\dagger$ is closed, we have $\mathcal{U}(s)^\dagger x \in D(\mathcal{A}^\dagger)$ and 
    \begin{align}
        \lim_{n\to\infty} \mathcal{A}^\dagger \mathcal{U}_n(s)^\dagger x 
        &=  \mathcal{A}^\dagger \mathcal{U}(s)^\dagger x \label{equ:proof_gle_3}
    \end{align}
    uniformly on bounded intervals. 
    
    For the last step of the proof, we note the following fact. Let $[0,t]\ni s\to f_n(s)\in H$ and $[0,t]\ni s\to g_n(s)\in H$ be sequences of continuous functions. Let $f_n \to f$ uniformly and $g_n\to g$ uniformly. Then, it follows from the Cauchy-Schwarz and triangle inequality that $(f_n(s),g_n(s)) \to (f(s),g(s))$ uniformly for $s\in [0,t]$. Hence, for all $x \in D(\mathcal{L}^\dagger)$, we have
    \begingroup
    \allowdisplaybreaks
    \begin{align*}
    &\left( \lim_{n\to\infty}\int^t_0   \mathcal{U}_n(s)(\mathcal{L}_n(s)-\mathcal{A})\mathcal{G}(t-s)\mathcal{Q}\mathcal{L}_n(t)z \, ds , x\right) \\
        &= \lim_{n\to\infty}\int^t_0 \left( \mathcal{U}_n(s)(\mathcal{L}_n(s)-\mathcal{A})\mathcal{G}(t-s)\mathcal{Q}\mathcal{L}_n(t)z ,  x \right) \, ds\\
        &= \lim_{n\to\infty}\int^t_0 \left( \mathcal{G}(t-s)\mathcal{Q}\mathcal{L}_n(t)z , (\mathcal{L}_n(s)^\dagger-\mathcal{A}^\dagger) \mathcal{U}_n(s)^\dagger x \right) \, ds\\
        &\overset{(\ref{equ:proof_gle_limit_of_ff})(\ref{equ:proof_gle_1})(\ref{equ:proof_gle_3})}{=}\int^t_0 \left( \mathcal{G}(t-s)\mathcal{Q}\mathcal{L}z , (\mathcal{L}^\dagger-\mathcal{A}^\dagger) \mathcal{U}(s)^\dagger x \right) \, ds\\
        &=\int^t_0 \left( \mathcal{G}(t-s)\mathcal{Q}\mathcal{L}z , (\mathcal{Q}\mathcal{L}^\dagger-\mathcal{A}^\dagger) \mathcal{U}(s)^\dagger x \right) \, ds\\
            &=\int^t_0 \left( \mathcal{G}(t-s)\mathcal{Q}\mathcal{L}z , \mathcal{Q}\mathcal{U}(s)^\dagger \mathcal{L}^\dagger x-\mathcal{A}^\dagger\mathcal{U}(s)^\dagger x \right) \, ds\\
        &\overset{(\ref{equ:proof_gle_2})(\ref{equ:proof_gle_3})}{=} \int^t_0\left( \mathcal{G}(t-s)\mathcal{Q}\mathcal{L}z ,  (\mathcal{U}(s)^\dagger x, z)(z,z)^{-1}\mathcal{Q}\mathcal{L}^\dagger z \right)   \, ds\\
        &= \int^t_0\left( \mathcal{G}(t-s)\mathcal{Q}\mathcal{L}z ,  \mathcal{Q}\mathcal{L}^\dagger z \right)(z,z)^{-1}  (\mathcal{U}(s)^\dagger x, z)^\ast \, ds\\
        &\overset{(\ref{def:gle_2:memory_kernel})}{=} \int^t_0 K(t-s)  ( \mathcal{U}(s) z, x) \, ds\\
        &= \left(\int^t_0 K(t-s) \mathcal{U}(s) z \, ds , x \right) \, .
    \end{align*}
    \endgroup
    Since $D(\mathcal{L}^\dagger)$ is dense in $H$, we conclude
    \begin{align*}
        \lim_{n\to\infty}\int^t_0   \mathcal{U}_n(s)(\mathcal{L}_n(s)-\mathcal{A})\mathcal{G}(t-s)\mathcal{Q}\mathcal{L}_n(t)z \, ds &= \int^t_0   K(t-s)\mathcal{U}(s)z \, ds \, .
    \end{align*}
    Inserting this into eq.~(\ref{equ:proof_gle_approximation_2}) yields
    \begin{align*}
        \frac{d}{dt}\mathcal{U}(t)z &=\mathcal{U}(t)\mathcal{P}\mathcal{L}z+ \mathcal{G}(t)\mathcal{Q}\mathcal{L}z+\int^t_0   K(t-s)\mathcal{U}(s)z \, ds \, , 
    \end{align*}
    for all $t\geq 0$. This concludes the proof.
\end{proof}

In the next section, we show that the GLE and 2FDT as well as all desired properties of the orthogonal dynamics can be proven by means of linear Volterra equations without relying on the bounded perturbation theorem.

\section{Alternative approach using Volterra equations}\label{ssec:volterra}

The following derivation of the GLE and 2FDT by means of the theory of Volterra equations is inspired by the existence proof for the orthogonal dynamics for finite-rank projections within the context of Hamiltonian vector fields given by Givon et al.~\cite[Theorem~4.1]{givon2005} in 2005. 

First, we observe that the memory kernel $K(t)$ as well as the fluctuating forces $\eta_t$ are uniquely determined by the GLE and 2FDT. 
Afterwards, it is left to restore the formulas for the memory kernel, fluctuating forces and orthogonal dynamics from sec.~\ref{ssec:semigroup_approach}.  

\begin{proposition}\label{proposition:gle}
    Let $z\in D(\mathcal{L})\cap D(\mathcal{L}^\dagger)$ with $z\neq 0$. 
    For $K\colon\mathbb{R}_+\to \mathbb{C}$ and $\eta\colon\mathbb{R}_+ \to H$, consider the system of equations
    \begin{align}
        \frac{d}{dt}\mathcal{U}(t)z &= \mathcal{U}(t)\mathcal{P}\mathcal{L}z+ \eta_t  + \int^t_0 K(t-s) \mathcal{U}(s)z \, ds \, , \label{equ:gle}  \\
        K(t) &= (\eta_t,\mathcal{Q}\mathcal{L}^\dagger z)(z,z)^{-1} \label{equ:fdt}   \, ,
     \end{align}
    for all $t\geq 0$. There exists a unique solution for $K$ and $\eta$. The memory kernel $K$ is given by the unique continuous solution of the Volterra equation
    \begin{align}
                K(t) &= (\mathcal{U}(t)\mathcal{QL}z,\mathcal{Q}\mathcal{L}^\dagger z)(z,z)^{-1}  \! - \! \int^t_0 K(t-s) (\mathcal{U}(s)z,\mathcal{Q}\mathcal{L}^\dagger z)(z,z)^{-1} ds \label{def:memory_kernel} \, .
    \end{align}
    The fluctuating forces $\eta$ are given by 
    \begin{align}
        \eta_t &= \mathcal{U}(t)\mathcal{QL}z - \int^t_0 K(t-s)\mathcal{U}(s)z \, ds \,  \label{def:fluctuating_forces} \, .
     \end{align}
\end{proposition}
\begin{proof}
    First, we show that $K$ and $\eta$ from eqs.~(\ref{def:memory_kernel}-\ref{def:fluctuating_forces}) solve the system of equations (\ref{equ:gle}-\ref{equ:fdt}). 
    Since $\{\mathcal{U}(t)\}_{t\geq 0}$ is strongly continuous, the orbit map $t\to \mathcal{U}(t)x$ is continuous for all $x\in H$. Hence, by the continuity of the scalar product, the coefficient functions of the Volterra equation $(\ref{def:memory_kernel})$ are continuous. Thus, there exists a unique continuous solution $K$ \cite[p.~25, Theorem~2.1.1]{burton}, i.e., $K$ is well-defined. Hence, the integral from eq.~(\ref{def:fluctuating_forces}) in the sense of Bochner \cite{Aliprantis2006, hille} is well-defined, and the integral interchanges with the scalar product \cite[p.~650, Theorem~8]{evans}. 

    Computing the time derivative of the orbit map $t\to\mathcal{U}(t)z$ immediately yields the \textit{generalized Langevin equation} (\ref{equ:gle})
    \begin{align*}
        \frac{d}{dt}\mathcal{U}(t)z &= \mathcal{U}(t)\mathcal{L}z \\
        &= \mathcal{U}(t)\mathcal{PL}z + \mathcal{U}(t)\mathcal{QL}z \\
        &\overset{(\ref{def:fluctuating_forces})}{=} \mathcal{U}(t)\mathcal{PL}z +\eta_t + \int^t_0 K(t-s) \mathcal{U}(s)z \, ds \, .
    \end{align*}
    Further, the \textit{second fluctuation dissipation theorem} (\ref{equ:fdt}) holds:
    \begin{align*}
        (\eta_t,\mathcal{Q}\mathcal{L}^\dagger z) &\overset{(\ref{def:fluctuating_forces})}{=} (\mathcal{U}(t)\mathcal{QL}z,\mathcal{Q}\mathcal{L}^\dagger z) - \left(\int^t_0 K(t-s)\mathcal{U}(s)z ds, \mathcal{Q}\mathcal{L}^\dagger z \right)  \\
    &= (\mathcal{U}(t)\mathcal{QL}z,\mathcal{Q}\mathcal{L}^\dagger z) - \int^t_0 K(t-s)(\mathcal{U}(s)z,\mathcal{Q}\mathcal{L}^\dagger z) ds \ \\
        &\overset{(\ref{def:memory_kernel})}{=} K(t)(z,z) \, .
    \end{align*}

    It is left to show that $K$ and $\eta$ are the only solution of the system of equations (\ref{equ:gle}-\ref{equ:fdt}). Clearly, the GLE (\ref{equ:gle}) holds if and only if $\eta$ is given by eq.~(\ref{def:fluctuating_forces}). Furthermore, inserting $\eta$ into eq.~(\ref{equ:fdt}) yields the Volterra equation (\ref{def:memory_kernel}). Hence, the assertion follows from the uniqueness of the solution of eq.~(\ref{def:memory_kernel}) \cite[p.~25, Theorem~2.1.1]{burton}.
\end{proof}

\begin{remark*}
    Note that proposition \ref{proposition:gle} immediately yields the GLE and 2FDT, except that it does not provide an explicit formula for the fluctuating forces in terms of the orthogonal dynamics. In fact, proposition \ref{proposition:gle} remains valid for arbitrary orthogonal projections. However, eq.~(\ref{equ:fdt}) corresponds to the expected form of the 2FDT only for rank-one projections. In any other case, the memory kernel and fluctuating forces as defined in proposition \ref{proposition:gle} will not be related to the orthogonal dynamics. 
\end{remark*}

Before we proceed, let us convince ourselves that the fluctuating forces $\eta_t$ from proposition \ref{proposition:gle} are indeed orthogonal to the observable of interest $z$. 

\begin{corollary}\label{corolloary:orthogonality}
    Under the assumptions of proposition \ref{proposition:gle}, we have
    \begin{align*}
        (\eta_t,z) &= 0 \, .
    \end{align*}
\end{corollary}
\begin{proof}
    Since $H$ is reflexive, the adjoint semigroup $\{\mathcal{U}(t)^\dagger\}_{t\geq 0}$ is strongly continuous with generator $\mathcal{L}^\dagger$ \cite{vanNeerven1992,engel}. Since $z\in D(\mathcal{L}^\dagger)$, for all $x\in H$ the map $t\to (\mathcal{U}(t)x,z)=(x,\mathcal{U}(t)^\dagger z)$ is continuously differentiable with derivative $t\to (\mathcal{U}(t)x,\mathcal{L}^\dagger z)$. Hence, we have
    \begin{align*}
        \frac{d}{dt}(\eta_t,z) &\overset{(\ref{def:fluctuating_forces})}{=} \frac{d}{dt} (\mathcal{U}(t)\mathcal{QL}z,z) - \frac{d}{dt}\int^t_0 K(s)(\mathcal{U}(t-s)z,z) ds \\
        &= (\eta_t,\mathcal{L}^\dagger z) -K(t)(z,z) \\
        &\overset{(\ref{equ:fdt})}{=} (\eta_t,\mathcal{L}^\dagger z)-(\eta_t,\mathcal{Q}\mathcal{L}^\dagger z) \\
        &= (\eta_t,\mathcal{P}\mathcal{L}^\dagger z) \\
        &= (\eta_t,z)(\mathcal{L}^\dagger z,z)^\ast(z,z)^{-1} \, ,    \end{align*}
    where we applied the Leibniz rule. This implies
    \begin{align*}
        (\eta_t,z) &= (\eta_0,z) e^{t (\mathcal{L}^\dagger z,z)^\ast(z,z)^{-1}} \, .
    \end{align*}
    However, since $\mathcal{Q}$ is the orthogonal projection onto $\text{span}(z)^\perp$, we have
    \begin{align*}
        (\eta_0,z) &= (\mathcal{Q}\mathcal{L}z,z) = 0 \, .
    \end{align*}
    Hence, we conclude $(\eta_t,z)=0$ for all $t\geq 0$. 
\end{proof}

Next, let us introduce the abstract Cauchy problem for the orthogonal dynamics. For $x\in H$ and $\mathbb{R}_+\ni t \to u(t) \in H$, the initial value problem
\begin{align}
    \begin{cases}
        \frac{d}{dt} u(t) = \overline{\mathcal{QL}}\mathcal{Q} u(t)\, , & t\geq0 \\ 
         u(0)=x \, , \\
    \end{cases} \label{def:initial_value_problem} 
\end{align}
is called abstract Cauchy problem associated to $\overline{\mathcal{QL}}\mathcal{Q}$ \cite[p.~145]{engel}.  

\begin{definition}[Mild solutions]
    Let $x\in H$. Suppose $\mathbb{R}_+\ni t \to u(x,t)\in H$  is continuous and satisfies
    \begin{align}
    \begin{cases}
        u(x,t) = x+\overline{\mathcal{QL}}\mathcal{Q}\int^t_0 u(x,s) \,  ds \, , & t\geq0, \\
           u(x,0)=x \,.
    \end{cases} \label{def:acp} 
    \end{align}
    Then, $u(x,\cdot)$ is called \emph{mild solution} of the abstract Cauchy problem associated to $\overline{\mathcal{QL}}\mathcal{Q}$ \cite[p.~146]{engel}. 
\end{definition}

Recall that the existence and uniqueness of mild solutions for all initial values $x\in H$ is guaranteed, provided that $\overline{\mathcal{QL}}\mathcal{Q}$ is the generator of a strongly continuous semigroup \cite[p.~146, Proposition~6.4]{engel}. Hence, we have to show that $\overline{\mathcal{QL}}\mathcal{Q}$ generates a strongly continuous semigroup that constitutes the desired orthogonal dynamics. In view of proposition \ref{proposition:gle}, it is clear that the orbit maps of the orthogonal dynamics can be defined via unique continuous solutions of linear Volterra equations. Afterwards, all desired properties of the orthogonal dynamics can be proven directly from its definition. 

\begin{theorem}\label{theorem:mild_solution}
    Let $z\in D(\mathcal{L}^\dagger)$ with $z\neq 0$ and $x\in H$. Let $\mathbb{R}_+\ni t \to f(x,t) \in \mathbb{C}$ be defined by the unique continuous solution of the Volterra equation 
    \begin{align}
        f(x,t) &= (\mathcal{U}(t)\mathcal{Q}x,\mathcal{Q}\mathcal{L}^\dagger z)(z,z)^{-1}-\int^t_0 f(x,t-s)(\mathcal{U}(s)z,\mathcal{Q}\mathcal{L}^\dagger z)(z,z)^{-1} \, ds \, . \label{def:f}
    \end{align}
    Let $\mathbb{R}_+ \ni t \to u(x,t) \in H$ be defined according to
    \begin{align}
                u(x,t) &:=\mathcal{P}x + \mathcal{U}(t)\mathcal{Q}x - \int^t_0 f(x,t-s)\mathcal{U}(s)z \, ds \, . \label{def:u}
    \end{align}
    Then, $u(x,\cdot)$ is a mild solution of the abstract Cauchy problem associated to $\overline{\mathcal{QL}}\mathcal{Q}$.  
\end{theorem}
\begin{proof}
Note that the Volterra equation (\ref{def:f}) has continuous coefficient functions. Thus, $f(x,\cdot)$ is well-defined via the unique continuous solution \cite[p.~25, Theorem~2.1.1]{burton}. In particular, $u(x,\cdot)$ is continuous. Furthermore, $f(x,\cdot)$ and $f(\mathcal{Q}x,\cdot)$ solve the same Volterra equation, i.e., $f$ does not depend on $\mathcal{P}x$. More explicitly, with $x_\perp := \mathcal{Q}x$ we have
\begin{align}
    f(x,t) &= f(x_\perp,t) \, , \\
    u(x,t) &= \mathcal{P}x + u(x_\perp,t) \, . \label{equ:u_equals}
\end{align}
Moreover, we have
\begin{align}
    f(x,t) &\overset{(\ref{def:f})(\ref{def:u})}{=} (u(x,t),\mathcal{Q}\mathcal{L}^\dagger z)(z,z)^{-1} \, . \label{equ:f_equals}
\end{align}
In the following, we require the derivative of $t\to (u(x_\perp,t),x')$, where $x'\in D(\mathcal{L}^\dagger)$. 
\begin{align}
    \frac{d}{dt} (u(x_\perp,t),x') &\overset{(\ref{def:u})}{=} \frac{d}{dt} (\mathcal{U}(t)x_\perp,x') - \frac{d}{dt}\int^t_0 f(x_\perp,s)(\mathcal{U}(t-s)z,x') \, ds \nonumber \\
    &= (\mathcal{U}(t)x_\perp,\mathcal{L}^\dagger x') - \int^t_0 f(x_\perp,s)(\mathcal{U}(t-s)z,\mathcal{L}^\dagger x') \, ds- f(x_\perp,t)(z,x') \nonumber \\
    &\overset{(\ref{def:u})}{=}  (u(x_\perp,t),\mathcal{L}^\dagger x')- f(x_\perp,t)(z,x') \label{equ:diff_orbit_map} \, .
\end{align}
In the same way as in the proof of corollary \ref{corolloary:orthogonality}, we show that $u(x_\perp,t) \in \mathcal{Q}H$: Since
    \begin{align*}
      (u(x_\perp,0),z) &= 0\, \intertext{and}
                \frac{d}{dt} (u(x_\perp,t),z) &\overset{(\ref{equ:diff_orbit_map})}{=} (u(x_\perp,t),\mathcal{L}^\dagger z) - f(x_\perp,t) (z,z) \\
        &\overset{(\ref{equ:f_equals})}{=} (u(x_\perp,t),\mathcal{L}^\dagger z) - (u(x_\perp,t),\mathcal{Q}\mathcal{L}^\dagger z) \\
        &= (u(x_\perp,t),\mathcal{P}\mathcal{L}^\dagger z) \\
        &= (u(x_\perp,t),z) (\mathcal{L}^\dagger z,z)^\ast (z,z)^{-1} \, , 
    \end{align*}
we conclude
    \begin{align}
        (u(x_\perp,t),z) &= 0  \, , \label{equ:u_remains_in_QH}
    \end{align}
for all $t\geq 0$. Hence, for all $y\in D(\mathcal{L}^\dagger\mathcal{Q})=\{y\in H: \mathcal{Q}y \in D(\mathcal{L}^\dagger)\}$, we have
    \begin{align*}
        \frac{d}{dt} (u(x_\perp,t),y) &\overset{(\ref{equ:u_remains_in_QH})}{=} \frac{d}{dt} (u(x_\perp,t),\mathcal{Q}y) \\
        &\overset{(\ref{equ:diff_orbit_map})}{=} (u(x_\perp,t),\mathcal{L}^\dagger \mathcal{Q} y)- f(x_\perp,t)(z,\mathcal{Q} y) \\
        &= (u(x_\perp,t),\mathcal{L}^\dagger \mathcal{Q} y) \, .
    \end{align*}
For all $y\in D(\mathcal{L}^\dagger\mathcal{Q})$, this implies
\begin{align*}
    (u(x_\perp,t),y) &= (u(x_\perp,0),y) + \left( \int^t_0 u(x_\perp,s) \,ds , \mathcal{L}^\dagger \mathcal{Q}  y\right) \, .
\end{align*}
According to proposition \ref{proposition:props_of_generator}, the operator $\mathcal{L}^\dagger\mathcal{Q}=\overline{\mathcal{QL}}^\dagger$ is densely defined and closed. Thus, it follows from a Lemma from John Ball \cite{ball} that $ \int^t_0 u(x_\perp,s) \,ds \in D(\overline{\mathcal{QL}})$ and 
\begin{align*}
    u(x_\perp,t) &= u(x_\perp,0) + \overline{\mathcal{QL}} \int^t_0 u(x_\perp,s) \,ds  \, .
\end{align*}
Since $\mathcal{Q}$ is bounded, $\mathcal{Q}$ interchanges with the Bochner integral \cite[p.~427, Lemma~11.45]{Aliprantis2006}. Since $u(x_\perp,t)=\mathcal{Q}u(x,t)$, we obtain
\begin{align*}
    u(x_\perp,t) &= u(x_\perp,0) + \overline{\mathcal{QL}}\mathcal{Q} \int^t_0 u(x,s) \,ds \, .
\end{align*}
Adding $\mathcal{P}x$ on both sides yields
\begin{align*}
    u(x,t) &\overset{(\ref{equ:u_equals})}{=} u(x,0) + \overline{\mathcal{QL}}\mathcal{Q} \int^t_0 u(x,s) \,ds \, ,
\end{align*}
for all $t\geq 0$. Therefore, $u(x,\cdot)$ is a mild solution of the abstract Cauchy problem associated to $\overline{\mathcal{QL}}\mathcal{Q}$. 
\end{proof}

For $z\in D(\mathcal{L}^\dagger)$ and $x\in H$, we have now established the fact that the orbit map of the orthogonal dynamics $u(x,\cdot)$ as defined in theorem \ref{theorem:mild_solution} is indeed a mild solution of the abstract Cauchy problem associated to $\overline{\mathcal{QL}}\mathcal{Q}$. Next, we show that $\overline{\mathcal{QL}}\mathcal{Q}$ generates the strongly continuous semigroup $\{u(\cdot,t)\}_{t\geq0}$. 

\begin{theorem}\label{theorem:strongly_continuous_semigroup}
    Let $z\in D(\mathcal{L}^\dagger)$ with $z\neq 0$ and let $u(x,t)$ be defined as in theorem \ref{theorem:mild_solution}. Then, $\{\mathcal{G}(t):=u(\cdot,t)\}_{t\geq0}$ is a strongly continuous semigroup of bounded linear operators on $H$.
\end{theorem}
\begin{proof}
    By theorem \ref{theorem:mild_solution}, the map $t\to u(x,t)$ is a mild solution for all $x\in H$. Thus, $\{u(\cdot,t)\}_{t\geq0}$ is strongly continuous. Let $x,y\in H$ and $\alpha\in \mathbb{C}$ arbitrary. It is easy to show that $t\to f(x,t)+\alpha f(y,t)$ and $t\to f(x+\alpha y,t)$ solve the same Volterra equation. Hence, by the uniqueness of a solution, we have for all $t\geq 0$
    \begin{align*}
        f(x+\alpha y,t) &= f(x,t)+\alpha f(y,t) \, .
    \end{align*}
    This implies that $\{u(\cdot,t)\}_{t\geq0}$ is a strongly continuous family of linear operators. We have to show that $u(\cdot,t)$ is bounded for all $t\geq 0$. There exist constants $M\geq 1$, $\omega_0\geq 0$ such that $\|\mathcal{U}(t)\| \leq Me^{\omega_0 t}$ for all $t\geq 0$ \cite[p.~39, Proposition~5.5]{engel}. Hence, for all $x\in H$ and $t\geq 0$, we have
    \begin{align*}
        \|u(x,t)\| &\overset{(\ref{def:u})}{=} \|\mathcal{P}x+\mathcal{U}(t)\mathcal{Q}x-\int^t_0 f(x,t-s)\mathcal{U}(s)z \, ds\| \\
        &\leq \|x\|+ \|\mathcal{U}(t)\mathcal{Q}x\| + \int^t_0 |f(x,t-s)| \|\mathcal{U}(s)z\| \, ds \\
        &\leq (1+Me^{\omega_0 t})\|x\| + Me^{\omega_0t}\|z\|\int^t_0 |f(x,s)| \, ds \\
        &\overset{(\ref{equ:f_equals})}{=} (1+Me^{\omega_0 t})\|x\| + Me^{\omega_0t}\|z\|\int^t_0 |(u(x,s),\mathcal{Q}\mathcal{L}^\dagger z)(z,z)^{-1}| \, ds \\ 
        &\leq (1+Me^{\omega_0 t})\|x\| + \frac{Me^{\omega_0t}\|\mathcal{L}^\dagger z\|}{\|z\|}\int^t_0 \|u(x,s)\| \, ds \\ 
        &\leq (1+Me^{\omega_0 t})\|x\| \exp\left\{ \frac{tMe^{\omega_0 t}\|\mathcal{L}^\dagger z\|}{\|z\| } \right\} \, , 
    \end{align*}
    where we used the Cauchy-Schwarz inequality in the penultimate line and Grönwall's inequality at the end. Thus, $\{u(\cdot,t)\}_{t\geq0}$ is a strongly continuous family of bounded linear operators on $H$. 
    
    It is left to show that $\{u(\cdot,t)\}_{t\geq 0}$ is a semigroup. Let $x_\perp\in \mathcal{Q}H$ arbitrary. Note that
    \begin{align*}
        \mathcal{U}(s)u(x_\perp,t) &\overset{(\ref{def:u})}{=} \mathcal{U}(s+t)x_\perp-\int^t_0 f(x_\perp,t-r)\mathcal{U}(s+r)z \, dr \\
        &= \mathcal{U}(s+t)x_\perp-\int^{t+s}_s f(x_\perp,t+s-r)\mathcal{U}(r)z \, dr \\
        &\overset{(\ref{def:u})}{=} u(x_\perp,t+s) +\int^s_0 f(x_\perp,t+s-r)\mathcal{U}(r)z \, dr \, , \intertext{hence}
        u(x_\perp,t+s) &= \mathcal{U}(s)u(x_\perp,t)-\int^s_0 f(x_\perp,t+s-r)\mathcal{U}(r)z \, dr \, .
    \end{align*}
    Since $u(x_\perp,t)\in\mathcal{Q}H$, this implies
    \begin{align}
        u(x_\perp,t+s) &\overset{(\ref{def:u})}{=} u(u(x_\perp,t),s) +\int^s_0 [f(u(x_\perp,t),s-r)-f(x_\perp,t+s-r)]\mathcal{U}(r)z \, dr \, . \label{equ:semigroup_property}
    \end{align}
    Furthermore, we have
    \begin{align*}
        &|f(u(x_\perp,t),s)-f(x_\perp,t+s)| \\
        &\quad \overset{(\ref{equ:f_equals})}{=} |(u(u(x_\perp,t),s)-u(x_\perp,t+s),\mathcal{QL^\dagger}z)(z,z)^{-1}| \\
        &\quad \leq \frac{\|\mathcal{L}^\dagger z\|}{\|z\|^2}\|u(u(x_\perp,t),s)-u(x_\perp,t+s)\| \\
        &\quad \overset{(\ref{equ:semigroup_property})}{\leq}  \frac{\|\mathcal{L}^\dagger z\|}{\|z\|^2} \int^s_0 |f(u(x_\perp,t),s-r)-f(x_\perp,t+s-r)| \|\mathcal{U}(r)z\| \, dr \\
        &\quad \leq \frac{Me^{\omega_0 s}\|\mathcal{L}^\dagger z\|}{\|z\|} \int^s_0 |f(u(x_\perp,t),r)-f(x_\perp,t+r)| \, dr \, .
    \end{align*}
    Again, Grönwall's inequality yields
    \begin{align*}
                |f(u(x_\perp,t),s)-f(x_\perp,t+s)| &= 0 \, ,
    \end{align*}
    for all $x_\perp\in \mathcal{Q}H$ and $t,s\geq 0$. Thus, eq.~(\ref{equ:semigroup_property}) simplifies to
    \begin{align}
        u(x_\perp,t+s) &=u(u(x_\perp,t),s) \, , \label{equ:semigroup_property_2}
    \end{align}
    for all $x_\perp\in \mathcal{Q}H$ and $t,s\geq 0$. Recall from eqs.~(\ref{equ:u_equals}), (\ref{equ:u_remains_in_QH}) that for all $x\in H$ and $t\geq 0$, we have $u(x,t)=\mathcal{P}x+u(\mathcal{Q}x,t)$, and $u(\mathcal{Q}x,t)\in \mathcal{Q}H$. This implies
    \begin{align*}
        u(x,t+s) &= \mathcal{P}x+u(\mathcal{Q}x,t+s) \\
         &\overset{(\ref{equ:semigroup_property_2})}{=} \mathcal{P}x+u(u(\mathcal{Q}x,t),s) \\
        &= \mathcal{P}u(x,t)+u(\mathcal{Q}u(x,t),s) \\
         &= u(u(x,t),s) \, ,
    \end{align*}
    for all $x \in H$ and $t,s\geq 0$. Thus, $\{u(\cdot,t)\}_{t\geq 0}$ is a strongly continuous semigroup of bounded linear operators on $H$.
\end{proof}

\begin{corollary}\label{corollary:generator}
    Let $z\in D(\mathcal{L}^\dagger)$ with $z\neq0$ and let $u(x,t)$ be defined as in theorem \ref{theorem:mild_solution}. Then, $\overline{\mathcal{QL}}\mathcal{Q} =\mathcal{QLQ}$ generates the strongly continuous semigroup $\{\mathcal{G}(t):=u(\cdot,t)\}_{t\geq 0}$.
\end{corollary}
\begin{proof}
    Let $\mathcal{A}$ be the generator of $\{u(\cdot,t)\}_{t\geq 0}$. Let $x\in D(\mathcal{A})$. Since $\overline{\mathcal{QL}}\mathcal{Q}$ is closed, and since $t\to u(x,t)$ is a continuous mild solution of the abstract Cauchy problem associated to $\overline{\mathcal{QL}}\mathcal{Q}$, we have
    \begin{align*}
        \mathcal{A}x &= \lim_{h \searrow 0} \frac{1}{h}[u(x,h)-x] \\
                    &=  \lim_{h \searrow 0} \overline{\mathcal{QL}}\mathcal{Q} \frac{1}{h} \int^h_0 u(x,s) \, ds \\
                    &= \overline{\mathcal{QL}}\mathcal{Q} \lim_{h \searrow 0} \frac{1}{h} \int^h_0 u(x,s) \, ds \\
                    &= \overline{\mathcal{QL}}\mathcal{Q} x \, . 
    \end{align*}
    This implies $\mathcal{A}x= \overline{\mathcal{QL}}\mathcal{Q}x$ for all $x\in D(\mathcal{A})$. By definition (\ref{def:u}) the orbit map $t\to u(x,t)$ is right differentiable at $t=0$ if and only if $x\in D(\mathcal{LQ})$. This implies $D(\mathcal{A})=D(\mathcal{LQ})$, hence $\mathcal{A}=\mathcal{QLQ}$. Since $\mathcal{QLP}$ is bounded and $\mathcal{A}$ is closed, we have 
        \begin{align*}
            \overline{\mathcal{QL}}\mathcal{Q} &= \overline{\mathcal{QL}}- \overline{\mathcal{QL}}\mathcal{P}  = \overline{\mathcal{QLQ}+\mathcal{QLP}}-\mathcal{QLP} = \overline{\mathcal{QLQ}} = \mathcal{A}\, ,
        \end{align*} 
        where we used the fact that $\overline{\mathcal{T}+\mathcal{S}}=\overline{\mathcal{T}}+\mathcal{S}$ if $\mathcal{T}$ is closable and $\mathcal{S}$ is bounded with domain $D(\mathcal{S})=H$, cf. \cite[p. 93]{weidmann}.
    Thus, $\overline{\mathcal{QL}}\mathcal{Q}=\mathcal{A}$ generates the strongly continuous semigroup $\{u(\cdot,t)\}_{t\geq 0}$. 
\end{proof}

\begin{corollary}\label{corollary:fluctuating_forces}
    Let $z\in D(\mathcal{L})\cap D(\mathcal{L}^\dagger)$ with $z\neq0$. Then, for all $t\geq 0$, the fluctuating forces $\eta_t$ from proposition \ref{proposition:gle} satisfy 
    \begin{align}
        \eta_t &= \mathcal{G}(t)\mathcal{QL}z  \, , \label{equ:cor_ff_1}  \\
         \eta_t &= \eta_0 + \overline{\mathcal{QL}}\mathcal{Q} \int^t_0 \eta_s \, ds \, ,  \label{equ:cor_ff_2}
    \end{align}
    where $\{\mathcal{G}(t)\}_{t\geq 0}$ is the strongly continuous semigroup generated by $\overline{\mathcal{QL}}\mathcal{Q}$. 
\end{corollary}
\begin{proof}
    The definition of $\eta_t$ from proposition \ref{proposition:gle} coincides with the definition of $u(\mathcal{QL}z,t)$ from theorem \ref{theorem:mild_solution}. By corollary \ref{corollary:generator}, we have $u(\mathcal{QL}z,t)=\mathcal{G}(t)\mathcal{QL}z$. This proves eq.~(\ref{equ:cor_ff_1}). Since $\{\mathcal{G}(t)\}_{t\geq 0}$ is a strongly continuous semigroup with generator $\overline{\mathcal{QL}}\mathcal{Q}$, the orbit map $t\to \eta_t=\mathcal{G}(t)\mathcal{QL}z$ is the unique mild solution of the abstract Cauchy problem associated to $\overline{\mathcal{QL}}\mathcal{Q}$ \cite[p.~146, Proposition~6.4]{engel}. This establishes eq.~(\ref{equ:cor_ff_2}). 
\end{proof}

We have now restored the missing formula for the fluctuating forces $\eta_t$ from proposition \ref{proposition:gle} in terms of the orthogonal dynamics. Hence, proposition \ref{proposition:gle} immediately yields the GLE and 2FDT.

\begin{theorem}[GLE and 2FDT]\label{theorem:gle}
    Let $z\in D(\mathcal{L})\cap D(\mathcal{L}^\dagger)$ with $z\neq0$. Then, the orbit map $\mathbb{R}_+ \ni t \mapsto \mathcal{U}(t)z \in D(\mathcal{L})$ solves the integro-differential equation
    \begin{align}
        \frac{d}{dt}\mathcal{U}(t)z &= \mathcal{U}(t)\mathcal{P}\mathcal{L}z + \mathcal{G}(t)\mathcal{Q}\mathcal{L}z+\int^t_0   K(t-s)\mathcal{U}(s)z \, ds \, , \\
        K(t)&:= (\mathcal{G}(t)\mathcal{Q}\mathcal{L}z,\mathcal{Q}\mathcal{L}^\dagger z)(z,z)^{-1}  \, ,
    \end{align}
    where $\{\mathcal{G}(t)\}_{t\geq 0}$ denotes the strongly continuous semigroup generated by $\overline{\mathcal{QL}}\mathcal{Q}$. 
\end{theorem}
\begin{proof}
    This is an immediate consequence of proposition \ref{proposition:gle} and corollary \ref{corollary:fluctuating_forces}.
\end{proof}

In particular, we have given an alternative proof of the GLE and 2FDT including the desired properties of the orthogonal dynamics and fluctuating forces without invoking the bounded perturbation theorem. Additionally, we will discuss unitary time evolution in the next subsection. 

Before we proceed, let us discuss how to apply the given results under time-reversal. 

\begin{remark*}[Time reversal]
    Suppose the time evolution is given by a strongly continuous group $\{\mathcal{U}(t)\}_{t\in\mathbb{R}}$. Then, $\{\mathcal{U}(\pm t)\}_{t\geq 0}$ is a strongly continuous semigroup with generator $\pm \mathcal{L}$ \cite[p.~79]{engel}. This implies that all results remain valid under the replacements $\mathcal{U}(t) \mapsto \mathcal{U}(-t)$ and $\mathcal{L}\mapsto -\mathcal{L}$. 
\end{remark*}

\begin{lemma}\label{lemma:time_reversal}
    Let $\{\mathcal{U}(t)\}_{t\in\mathbb{R}}$ be a strongly continuous group and $z \in D(\mathcal{L}^\dagger)$ with $z\neq0$. Then, $\overline{\mathcal{QL}}\mathcal{Q}$ generates a strongly continuous group. 
\end{lemma}
\begin{proof}
    Since $\{\mathcal{U}(\pm t)\}_{t\geq 0}$ is a strongly continuous semigroup with generator $\pm \mathcal{L}$ \cite[p.~79]{engel}, it follows from corollary \ref{corollary:generator} that the operator $\pm \overline{\mathcal{QL}}\mathcal{Q}$ generates a strongly continuous semigroup $\{\mathcal{G}(\pm t)\}_{t\geq 0}$. This implies that $\{\mathcal{G}(t)\}_{t\in \mathbb{R}}$ is a strongly continuous group with generator $\overline{\mathcal{QL}}\mathcal{Q}$ \cite[p.~79]{engel}.
\end{proof}

\subsection{Unitary time evolution}\label{ssec:unitary_time_evolution}

Suppose that the time evolution of observables is given by a strongly continuous unitary group $\{\mathcal{U}(t)\}_{t\in \mathbb{R}}$. This includes, e.g., large classes of classical systems in thermal equilibrium. Then, by Stone's theorem for one-parameter unitary groups \cite{stone1932,neumann1932}, the generator $\mathcal{L}$ is skew-adjoint \cite[p.~210, Theorem~10.15]{hall2013}. Note that the converse statement is also true \cite[p.~208, Proposition~10.14]{hall2013}. In this case, the fluctuating forces are stationary. Moreover, the orthogonal dynamics is a strongly continuous unitary group.

\begin{theorem}\label{theorem:unitary_group}
    Let $\mathcal{L}$ be skew-adjoint and $ z \in D(\mathcal{L})$ with $z\neq 0$. Then, $\overline{\mathcal{QL}}\mathcal{Q}$ generates a strongly continuous unitary group $\{\mathcal{G}(t)\}_{t\in\mathbb{R}}$. 
\end{theorem}
\begin{proof}
    By corollary \ref{corollary:generator} and proposition \ref{proposition:props_of_generator}, we have $\overline{\mathcal{QL}}\mathcal{Q}=\mathcal{QLQ}$ and $[\mathcal{L}^\dagger\mathcal{Q}]^\dagger = \overline{\mathcal{QL}}$. Thus, 
    \begin{align}
        (\overline{\mathcal{QL}}\mathcal{Q})^\dagger &= (\mathcal{QLQ})^\dagger 
        = [\mathcal{LQ}]^\dagger \mathcal{Q} 
        = -[\mathcal{L^\dagger Q}]^\dagger \mathcal{Q} 
        = -\overline{\mathcal{QL}}\mathcal{Q} \, ,
    \end{align}
    where we used \cite[p.~330, Theorem~13.2]{rudin1974}.
\end{proof}

\begin{corollary}\label{corollary:stationary}
    Let $\mathcal{L}$ be skew-adjoint and $ z \in D(\mathcal{L})$  with $z\neq0$. Then, for all $r,s,t \in\mathbb{R}$, 
    \begin{align*}
        (\eta_{t+r},\eta_{s+r}) &= (\eta_t,\eta_s) \, ,
    \end{align*}
    where $\eta_t:=\mathcal{G}(t)\mathcal{QL}z$ and $\{\mathcal{G}(t)\}_{t\in\mathbb{R}}$ is the strongly continuous unitary group generated by $\overline{\mathcal{QL}}\mathcal{Q}$. 
\end{corollary}

Note that the identity $\overline{\mathcal{QL}}\mathcal{Q}=\overline{\mathcal{QLQ}}$ does not hold for arbitrary orthogonal projections $\mathcal{Q}$ and skew-adjoint generators $\mathcal{L}$. A notable counterexample for an infinite rank projection (rather than the Mori projection) is given by Givon et al.~\cite{givon2005}. We now present this result for the reader's convenience.

\begin{example}\label{example:givon}
    Let $H:=L^2(\mathbb{R})$, $\mathcal{L}:=\partial_x$, $D(\mathcal{L}):=W^{1,2}(\mathbb{R})$, $\mathcal{Q}f(x):=\Theta(x)f(x)$, where $\mathcal{L}$ is the weak derivative and $\Theta$ denotes the Heaviside function. Then, 
    \begin{align*}
        \overline{\mathcal{QLQ}} \subsetneq  -\overline{\mathcal{QLQ}}^\dagger = \overline{\mathcal{QL}}\mathcal{Q} \, ,
    \end{align*}
    and $\overline{\mathcal{QL}}\mathcal{Q}$ does \textit{not} generate a strongly continuous unitary group. 
\end{example}
\begin{proof}
    The proof is given in \ref{app:example}.
\end{proof}

\section{Classical statistical mechanics}\label{sec:classical_statistical_mechanics}

In this section, we demonstrate how the previous results can be applied to statistical ensembles of (possibly non-Hamiltonian) classical systems. Let the state space be given by $\Omega=\mathbb{R}^n$ and $\mu$ be the Lebesgue measure on $(\Omega,\Sigma)$, where $\Sigma$ denotes the Lebesgue $\sigma$-algebra. Let $\rho \in L^1(\mu)$ be non-negative and normalized. The space of real-valued observables $L^2(P)$ is the set of square-integrable random variables for the probability measure $P$,
\begin{align}
    P(\mathcal{M}) &:= \int \rho\mathbf{1}_\mathcal{M}  \, d\mu \, , \quad \mathcal{M}\in \Sigma \, , \label{def:probability_measure}
\end{align}
where $\mathbf{1}_\mathcal{M}$ denotes the indicator function. The space $L^2(P)$ is a real Hilbert space \cite[p.~153, Corollary~7.22]{klenke} equipped with the scalar product
\begin{align}
    (x,y)_{L^2(P)}&:= \int xy  \,dP =  \int xy\rho \, d\mu \, . \label{def:scalar_product_L2P}
\end{align}
The complex Hilbert space of observables $H$ is introduced via the complexification of $L^2(P)$, i.e.,
\begin{align*}
    H &:= \mathbb{C} \otimes L^2(P) \, , \\
    (\alpha x, \beta y) &:= \alpha\beta^\ast \cdot (x,y)_{L^2(P)} \, , 
\end{align*}
for all $x,y\in L^2(P)$ and $\alpha,\beta \in \mathbb{C}$, where $\otimes$ denotes the tensor product. 

Now, let $\mathbf{F}\colon\Omega\to\Omega$ be continuously differentiable with bounded derivative. Then, there exists a continuous flow $\varphi \colon \mathbb{R}\times \Omega \mapsto \Omega$, such that \cite[p.~91]{engel}
\begin{align*}
    \varphi_0 &= \text{id}_\Omega \, , \\
    \varphi_{t+s} &= \varphi_{t} \circ \varphi_{s}
    \, , \\
   \frac{d}{dt} \varphi_{t}(\mathbf{x}) &= \mathbf{F}\circ \varphi_t(\mathbf{x}) \, ,
\end{align*}
for all $t,s\in\mathbb{R}$ and $\mathbf{x}\in\Omega$. Let $X$ be the Banach space of continuous functions that vanish at infinity,
\begin{align*}
    X &:= C_0(\Omega) \, \intertext{equipped with the norm}
    \| x \|_X &:= \sup_{\mathbf{x} \in \Omega}|x(\mathbf{x})| \, .
\end{align*}
The composition operator (Koopman operator \cite{koopman31}) 
\begin{align}
    \mathcal{K}(t)x&:=x\circ\varphi_t \, , \quad x\in X \, , \label{def:koopman}
\end{align}
is a strongly continuous group on $X$, where the generator is given by the closure of $(\mathbf{F}\cdot\boldsymbol{\nabla}, C^1_c(\Omega))$ in $X$, and $C^1_c(\Omega)$ is a core of the generator \cite[p.~91-92]{engel}.

\begin{theorem}\label{theorem:time_evolution}
    Let $\rho \in L^1(\mu)$ be non-negative and normalized. Let $\mathbf{F}\colon\Omega\to\Omega$ be continuously differentiable with bounded derivative. We further assume $\rho\mathbf{F}\in W^{1,1}(\Omega)$ and $\omega_0:=\frac{1}{2}\|\rho^{-1}\textup{div}(\rho\mathbf{F})\|_{L^\infty(P)}<\infty$. (We note that the set where $\rho$ vanishes is a nullset for $P$, hence $\rho^{-1}$ is well-defined $P$-almost everywhere.) Then
    \begin{itemize}
        \item $\overline{\mathbf{F}\cdot\boldsymbol{\nabla}}^{L^2(P)}$ generates a strongly continuous group $\{\mathcal{U}(t)\}_{t\in\mathbb{R}}$ on $L^2(P)$,
        \item $\|\mathcal{U}(t)\| \leq e^{\omega_0 |t|}$ for all $t\in\mathbb{R}$,
        \item $\mathcal{U}(t)x =  \mathcal{K}(t)x$ for all $x\in C^1_c(\Omega)$ and $t\in\mathbb{R}$.
    \end{itemize}
    If in addition $\textup{div}(\rho\mathbf{F})=0$, then
    \begin{itemize}
        \item $\mathcal{U}(t)$ is a linear isometry on $L^2(P)$ for all $t\in\mathbb{R}$,
        \item $P$ is invariant under the map $\varphi_t$ for all $t\in\mathbb{R}$,
        \item $\mathcal{U}(t)x= x\circ \varphi_t$ for all $x\in L^2(P)$ and $t\in\mathbb{R}$.
    \end{itemize}
\end{theorem}
\begin{proof}
The proof is given in \ref{app:time_evolution}.
\end{proof}

Under the assumptions of theorem \ref{theorem:time_evolution},
it follows from Hölder's inequality that $\{\mathcal{U}(t)\}_{t\in\mathbb{R}}$ is a strongly continuous group such that $\|\mathcal{U}(t)\| \leq e^{\omega_0 |t|}$ for all $t\in\mathbb{R}$. The constant $\omega_0=\frac{1}{2}\|\rho^{-1}\textup{div}(\rho\mathbf{F})\|_{L^\infty(P)}$ is given by
\begin{equation}
    \omega_0 = \esssup_{\text{supp}(\rho)} \left\lvert \frac{1}{2\rho}\text{div}(\rho\mathbf{F}) \right\rvert \, . \label{eq:omega0}
\end{equation}
Let $\omega_\pm \in\mathbb{R}$ denote the smallest constant such that $\|\mathcal{U}(\pm t)\|\leq e^{\omega_\pm t}$ for all $t\geq 0$. According to Zhu et al. \cite[Eq.~(17)-(18)]{zhu18}, $\omega_\pm$ is given by
\begin{align}
    \omega_\pm &= \esssup_{\text{supp}(\rho)} \left(  \pm \frac{ 1}{2\rho}\text{div}(\rho \mathbf{F}) \right)\, ,
\end{align}
which is a consequence of the growth bounds for strongly continuous semigroups from Davies \cite[Lemma 2.2]{davies}. In the article by Davies, it is assumed that the semigroup is strongly continuous. In contrast, under the assumptions of theorem~\ref{theorem:time_evolution}, the strong continuity follows from Hölder's inequality. Since $\omega_0=\max\{\omega_+,\omega_-\}$, it follows from the work by Zhu et al. and Davies that $\omega_0$ is the smallest constant such that $\|\mathcal{U}(t)\|\leq e^{\omega_0 |t|}$ for all $t\in\mathbb{R}$.     

For simplicity, let us outline the main results for the case of stationary classical systems for which $\mathcal{L}$ is skew-adjoint. Under the assumptions of theorem \ref{theorem:time_evolution} with $\text{div}(\rho\mathbf{F})=0$, the time evolution of observables is described by the strongly continuous unitary group $\{\mathcal{U}(t)\}_{t\in\mathbb{R}}$ generated by $\mathcal{L}:=\overline{\mathbf{F}\cdot\boldsymbol{\nabla}}^H$. In this case, for any observable $0\neq z\in D(\mathcal{L})$, we obtain
\begin{align*}
    \frac{d}{dt} [z\circ\varphi_t] &= \omega [z\circ\varphi_t] + \eta_t+\int^t_0   K(t-s) [z\circ\varphi_s] \, ds \, ,  \intertext{with}
    \omega &:= (\mathcal{L}z,z)(z,z)^{-1}  \, ,\\
    \eta_t &:= \mathcal{G}(t)\mathcal{Q}\mathcal{L}z  \, ,  \\
    K(t)&:= -(\eta_t,\eta_0)(z,z)^{-1} \, , 
\end{align*}
where $\{\mathcal{G}(t)\}_{t\in \mathbb{R}}$ is the strongly continuous unitary group generated by $\overline{\mathcal{QL}}\mathcal{Q}=\mathcal{QLQ}$ and $\mathcal{Q}$ is the orthogonal projection onto $\text{span}(z)^\perp$. Note that the drift term $\omega$ vanishes for real observables, since $\mathcal{L}$ is skew-adjoint. Further, we have $(\eta_t,z)=0$, and $(\eta_{t+r},\eta_{s+r})=(\eta_t,\eta_s)$. In addition, the orthogonal dynamics is given by $\mathcal{G}(t)=\mathcal{P} + \mathcal{Q}e^{\mathcal{LQ}t}$, where $\{e^{\mathcal{LQ}t}\}_{t\in \mathbb{R}}$ denotes the strongly continuous group generated by $\mathcal{LQ}$. 

\section{Non-autonomous systems}\label{sec:non_autonomous_systems}
Finally, we ask whether the GLE also holds for systems that evolve under time-dependent dynamics such as systems subjected to time-dependent external forces or active matter.

For a time-dependent flow $\varphi \colon \mathbb{R}^2 \times \Omega \to \Omega$, where $\Omega$ denotes some configuration space, we have $\varphi(t,r)\circ\varphi(r,s)=\varphi(t,s)$ and $\varphi(t,t)=\text{id}_\Omega$, for all $t,s,r\in\mathbb{R}$. If $H$ is a complex Hilbert space of observables such that $x\colon\Omega\to\mathbb{C}$ for all $x\in H$, then the time evolution of observables is given by the composition operator $\mathcal{U}(t,s)x:=x\circ \varphi(t,s)$. Hence, due to the composition, $\{\mathcal{U}(t,s)\}_{t,s\in\mathbb{R}}$ is expected to resemble an evolution family with \textit{reversed} order of its `group properties'. More explicitly, we have
\begin{align*}
    \mathcal{U}(t,s)x &= x\circ \varphi(t,s) \\
    &= x\circ \varphi(t,r)\circ\varphi(r,s) \\
    &= \mathcal{U}(r,s)[x\circ \varphi(t,r)] \\
    &= \mathcal{U}(r,s)\mathcal{U}(t,r)x \, .
\end{align*}
This corresponds to the notion of \textit{negatively time-ordered exponentials}. 

\begin{definition}
    Let $\{\mathcal{U}(t,s)\}_{t,s\in\mathbb{R}}$ be a strongly continuous family of bounded linear operators on $H$ such that for all $r,s,t\in \mathbb{R}$
    \begin{align}
        \mathcal{U}(t,t) &= 1 \, , \\
        \mathcal{U}(r,s)\mathcal{U}(t,r) &= \mathcal{U}(t,s) \, . \label{equ:group_property}
    \end{align}
    For all $t\in\mathbb{R}$, let $\mathcal{L}(t)$ be defined by
    \begin{align}
        \mathcal{L}(t)x &:= \lim_{h\to 0} \frac{1}{h}[\mathcal{U}(t+h,t)x-x] \label{def:L_t} \intertext{with domain}
        D(\mathcal{L}(t)) &:= \{x\in H: \mathcal{L}(t)x \in H \} \, .
    \end{align}
    Further, let the subspace $D_\mathcal{L}$ be defined by
    \begin{align*}
        D_\mathcal{L} &:= \bigcap_{t\in\mathbb{R}} D(\mathcal{L}(t)) \, .
    \end{align*}
\end{definition}

Note in particular that for all $t,s\in \mathbb{R}$ and $x\in D_\mathcal{L}$
\begin{align}
    \frac{d}{dt} \mathcal{U}(t,s)x &\overset{(\ref{equ:group_property})}{=} \lim_{h\to 0} \frac{1}{h} \mathcal{U}(t,s)[\mathcal{U}(t+h,t)-1]x \\
    &\overset{(\ref{def:L_t})}{=} \mathcal{U}(t,s)\mathcal{L}(t)x \, . \label{equ:time_derivative}
\end{align}
For simplicity, we restrict ourselves to the Mori projection. 
\begin{definition}
For some fixed observable of interest $0\neq z \in H$ and some fixed initial time $t_0\in\mathbb{R}$, let the Mori projection operator $\mathcal{P}(t)$ be given by
\begin{align}
        \mathcal{P}(t) &:= (\cdot,z)_t(z,z)_t^{-1} z  \, , \label{def:ns_projection} \\
        \mathcal{Q}(t) &:= 1-\mathcal{P}(t) \, , 
\end{align}
where we introduced the ``time-dependent'' scalar product
\begin{align}
            (x,y)_t &:= (\mathcal{U}(t,t_0)x,\mathcal{U}(t,t_0)y) \label{def:ns_scalar_product} \, .
\end{align}
\end{definition}

The following proposition is the analogue of proposition \ref{proposition:gle} for non-autonomous systems. 

\begin{proposition}\label{proposition:nsgle}
    Let $t\to \mathcal{U}(t,t_0)\mathcal{L}(t)z$ be continuous for some $z \in D_\mathcal{L}$  with $z\neq0$. For $K\colon\mathbb{R}^2\to\mathbb{C}$ and $\eta\colon\mathbb{R}^2\to H$, consider the system of equations 
    \begin{align}
        \frac{d}{dt} \mathcal{U}(t,t_0)z &= \mathcal{U}(t,t_0)\mathcal{P}(t)\mathcal{L}(t)z+\mathcal{U}(s,t_0)\eta_{ts}+\int^t_s K(t,r)\mathcal{U}(r,t_0)z \, dr \, , \label{equ:nsGLE}\\
        K(t,s) &= -(\eta_{ts},\eta_{ss})_s(z,z)_s^{-1} \, , \label{equ:ns2FDT}
    \end{align}
    for all $t,s \in \mathbb{R}$. There exists a unique solution for $K$ and $\eta$. We call~\eqref{equ:nsGLE} \emph{non-stationary Generalized Langevin Equation}, abbreviated  nsGLE. For all $t\in\mathbb{R}$, the memory kernel $K(t,\cdot)\colon\mathbb{R}\to \mathbb{C}$ is given by the unique continuous solution of the Volterra equation
    \begin{align}
            K(t,s)  &= -(\mathcal{U}(t,s)\mathcal{Q}(t)\mathcal{L}(t)z,\mathcal{Q}(s)\mathcal{L}(s)z)_s(z,z)_s^{-1} \notag \\
                    &\,\,\,\quad + \int^t_s K(t,r)(\mathcal{U}(r,s)z,\mathcal{Q}(s)\mathcal{L}(s)z)_s(z,z)_s^{-1} \, dr \, .
       \label{def:nsk} 
    \end{align}
    The fluctuating forces $\eta\colon\mathbb{R}^2 \to H$ are given by
    \begin{align}
        \eta_{ts} &= \mathcal{U}(t,s)\mathcal{Q}(t)\mathcal{L}(t)z-\int^t_s K(t,r) \mathcal{U}(r,s)z \, dr \, . \label{def:nsff}
    \end{align}
\end{proposition}
\begin{proof}
    Under the given assumptions, one easily verifies that for any fixed $t\in\mathbb{R}$, the coefficient functions of the Volterra equation (\ref{def:nsk}) are continuous. Hence, $K(t,\cdot)$ is well-defined by the unique continuous solution of the Volterra equation (\ref{def:nsk}) \cite[p.~25, Theorem~2.1.1]{burton}. Further, the vector-valued integrals are well-defined and interchange with the scalar product. Computing the time derivative of $t\to\mathcal{U}(t,t_0)z$ immediately yields the nsGLE, eq.~(\ref{equ:nsGLE}),
    \begin{align*}
        \frac{d}{dt}\mathcal{U}(t,t_0)z &\overset{(\ref{equ:time_derivative})}{=} \mathcal{U}(t,t_0)\mathcal{L}(t)z \\
        &=\mathcal{U}(t,t_0)\mathcal{P}(t)\mathcal{L}(t)z+\mathcal{U}(t,t_0)\mathcal{Q}(t)\mathcal{L}(t)z \\
        &\overset{(\ref{equ:group_property})}{=} \mathcal{U}(t,t_0)\mathcal{P}(t)\mathcal{L}(t)z+\mathcal{U}(s,t_0)\mathcal{U}(t,s)\mathcal{Q}(t)\mathcal{L}(t)z  \\
        &\overset{(\ref{def:nsff})}{=} \mathcal{U}(t,t_0)\mathcal{P}(t)\mathcal{L}(t)z+\mathcal{U}(s,t_0)\eta_{ts} + \int^t_s K(t,r)\mathcal{U}(r,t_0)z \, dr \, , 
    \end{align*}
    where we used the group properties in the last two steps. Moreover, inserting $\eta_{ts}$ into eq.~(\ref{equ:ns2FDT}) yields eq.~(\ref{def:nsk}). Hence, eq.~(\ref{equ:ns2FDT}) holds. 

    It is left to show that $K$ and $\eta$ are uniquely determined by eqs.~(\ref{equ:nsGLE}) and (\ref{equ:ns2FDT}). Clearly, eq.~(\ref{equ:nsGLE}) holds if and only if $\eta_{ts}$ is given by eq.~(\ref{def:nsff}). Furthermore, for any fixed $t\in\mathbb{R}$, inserting $\eta_{ts}$ into eq.~(\ref{equ:ns2FDT}) yields the Volterra equation (\ref{def:nsk}). Hence, the assertion follows from the uniqueness of the solution of eq.~(\ref{def:nsk}). 
\end{proof}

Similar to corollary \ref{corolloary:orthogonality}, we convince ourselves that the fluctuating forces $\eta_{ts}$ as defined in proposition \ref{proposition:nsgle} are indeed orthogonal to the observable of interest $z$ w.r.t. the scalar product $(\cdot,\cdot)_s$. As a consequence, the non-stationary version of the 2FDT holds. 

\begin{corollary}[Orthogonality and 2FDT]\label{corollary:2fdt}
    Under the assumptions of proposition \ref{proposition:nsgle}, we have
    \begin{align*}
        (\eta_{ts},z)_s &= 0 \, , \\
        K(t,s) &= -(\eta_{tt_0},\eta_{st_0})(z,z)_s^{-1} \, .
    \end{align*}
\end{corollary}
\begin{proof}
    Taking the derivative of eq.~(\ref{equ:nsGLE}) w.r.t. $s$ yields
    \begin{align}
        \frac{d}{ds} \mathcal{U}(s,t_0)\eta_{ts} &= K(t,s)\mathcal{U}(s,t_0)z \, . \label{equ:proof_ns_0}
    \end{align}
    With this, we obtain
    \begin{align*}
        \frac{d}{ds}(\eta_{ts},z)_s &\overset{(\ref{def:ns_scalar_product})}{=} \frac{d}{ds} (\mathcal{U}(s,t_0)\eta_{ts},\mathcal{U}(s,t_0)z) \\
        &\overset{(\ref{equ:proof_ns_0})(\ref{equ:time_derivative})(\ref{def:ns_scalar_product})}{=} K(t,s)(z,z)_s + (\eta_{ts},\mathcal{L}(s)z)_s \\
        &\overset{(\ref{equ:ns2FDT})}{=} -(\eta_{ts},\eta_{ss})_s + (\eta_{ts},\mathcal{L}(s)z)_s \\
        &\overset{(\ref{def:nsff})}{=} -(\eta_{ts},\mathcal{Q}(s)\mathcal{L}(s)z)_s + (\eta_{ts},\mathcal{L}(s)z)_s \\
        &= (\eta_{ts},\mathcal{P}(s)\mathcal{L}(s)z)_s \\
        &\overset{(\ref{def:ns_projection})}{=} (\eta_{ts},z)_s(\mathcal{L}(s)z,z)_s^\ast(z,z)_s^{-1} \, .
    \end{align*}
    This implies
    \begin{align*}
         (\eta_{ts},z)_s &= (\eta_{tt},z)_t\exp\left\{\int^s_t(\mathcal{L}(r)z,z)_r^\ast(z,z)_r^{-1} \, dr\right\} \, .
    \end{align*}
    The initial value is given by
    \begin{align*}
        (\eta_{tt},z)_t &\overset{(\ref{def:nsff})}{=} (\mathcal{Q}(t)\mathcal{L}(t)z,z)_t \\
        &= (\mathcal{L}(t)z,z)_t-(\mathcal{P}(t)\mathcal{L}(t)z,z)_t \\
        &\overset{(\ref{def:ns_projection})}{=} (\mathcal{L}(t)z,z)_t-(\mathcal{L}(t)z,z)_t(z,z)_t^{-1}(z,z)_t \\
        &= 0 \, .
    \end{align*}
    Hence, we obtain the orthogonality relation
    \begin{align}
        (\eta_{ts},z)_s &= 0 \, . \label{equ:proof_ns_orthogonality}
    \end{align}
    Now, let us show that $r\to(\eta_{tr},\eta_{sr})_r$ is constant.
    \begin{align*}
        \frac{d}{dr} (\eta_{tr},\eta_{sr})_r &\overset{(\ref{def:ns_scalar_product})}{=} \frac{d}{dr}(\mathcal{U}(r,t_0)\eta_{tr},\mathcal{U}(r,t_0)\eta_{sr}) \\
        &\overset{(\ref{equ:proof_ns_0})}{=} K(t,r)(\mathcal{U}(r,t_0)z,\mathcal{U}(r,t_0)\eta_{sr})+K(s,r)^\ast(\mathcal{U}(r,t_0)\eta_{tr},\mathcal{U}(r,t_0)z) \\
        &\overset{(\ref{def:ns_scalar_product})}{=}  K(t,r)(z,\eta_{sr})_r+K(s,r)^\ast(\eta_{tr},z)_r \\
        &\overset{(\ref{equ:proof_ns_orthogonality})}{=} 0 \, .
    \end{align*}
    Hence, we have
    \begin{align*}
        K(t,s) &\overset{(\ref{equ:ns2FDT})}{=} -(\eta_{ts},\eta_{ss})_s(z,z)_s^{-1} \\
        &= -(\eta_{tr},\eta_{sr})_r(z,z)_s^{-1} \, .
    \end{align*}
    For $r=t_0$, this implies
    \begin{align*}
        K(t,s) &= -(\eta_{tt_0},\eta_{st_0})(z,z)_s^{-1} \, .
    \end{align*}
    This concludes the proof. 
\end{proof}

Proposition \ref{proposition:nsgle} and corollary \ref{corollary:2fdt} establish the nsGLE and 2FDT except that we cannot provide a rigorous expression for the fluctuating forces in terms of the orthogonal dynamics. This would require the treatment of non-autonomous abstract Cauchy problems, which are generally much more difficult than autonomous abstract Cauchy problems \cite{engel}.

\section{Discussion and Conclusion}\label{sec:conclusion}

In general, the scope of validity of the Mori-Zwanzig projection operator technique is a subtle issue. In order to derive the GLE, one typically makes use of the Dyson formula  
\begin{align}
    \mathcal{G}(t)x &= \mathcal{U}(t)x - \int^t_0 \mathcal{U}(t-s) \mathcal{PL} \mathcal{G}(s)x \,  ds   \, , \label{equ:dyson_formula}
\end{align}
where $\{\mathcal{G}(t)\}_{t\geq 0}$ denotes the evolution operator for the orthogonal dynamics. The main difficulty regarding a rigorous mathematical treatment for infinite-rank projections is the absence of a reasonable definition for the orthogonal dynamics $\{\mathcal{G}(t)\}_{t\geq 0}$. 

\begin{remark*}
    If $\mathcal{PL}$ is bounded and if $\{\mathcal{G}(t)\}_{t\geq 0}$ is a strongly continuous semigroup with generator $\mathcal{QL}=\mathcal{L}-\mathcal{PL}$, the Dyson identity (\ref{equ:dyson_formula}) equals the variation of constants formula obtained from the bounded perturbation theorem \cite[p.~161, Corollary~1.7]{engel}. However, $\mathcal{PL}$ is typically unbounded. Moreover, the perturbed operator $\mathcal{QL}$ might fail to be closed, in which case $\mathcal{QL}$ does not generate a strongly continuous semigroup, cf. \cite{givon2005} and example \ref{example:givon}. Hence, we emphasize the results by John Ball \cite{ball}.
\end{remark*}

The Dyson formula can be viewed as an application of John Ball's version of the variation of constants formula \cite{ball}, cf.~sec.~\ref{ssec:semigroup_approach}. This requires that $t\to \mathcal{PL}\mathcal{G}(t)x$ is continuous and $t\to \mathcal{G}(t)x$ has to be a weak solution of
\begin{align}
    \frac{d}{dt}\mathcal{G}(t)x &= \mathcal{L}\mathcal{G}(t)x - \mathcal{PL}\mathcal{G}(t)x 
\end{align}
in the sense of Ball \cite{ball}. This is the case, for instance, if $x\in D(\mathcal{L})\cap \mathcal{Q}H$ and $\mathcal{G}(t)x:= \mathcal{Q}e^{\mathcal{LQ}t}x$, where we assume that $\mathcal{LQ}$ generates a strongly continuous semigroup. In general, it turns out to be extremely difficult to state for which systems and projections this assumption holds. 

We have therefore focussed on the simpler case of rank-one projections (Mori's projection). A similar treatment can be carried out for finite-rank projections. Let the time-evolution operator $\{\mathcal{U}(t)\}_{t\geq 0}$ be a strongly continuous semigroup on a complex Hilbert space $H$ with generator $\mathcal{L}$. The projection operator be given by $\mathcal{P}:= (.,z)(z,z)^{-1}z$, where $(.,.)$ denotes the scalar product with complex conjugation in its second argument and $z\neq 0$ denotes the observable of interest. Its complementary projection is denoted by $\mathcal{Q}:=1-\mathcal{P}$. Under these assumptions, we have proven the following statements in two independent ways.
\begin{enumerate}
    \item If $z\in D(\mathcal{L}^\dagger)$, then $\overline{\mathcal{QL}}\mathcal{Q}$ generates a strongly continuous semigroup, denoted by $\{\mathcal{G}(t)\}_{t\geq 0}$. 
    \item If $z\in D(\mathcal{L})\cap D(\mathcal{L}^\dagger)$, the GLE takes the form
\begin{align*}
    \frac{d}{dt}\mathcal{U}(t)z &= \mathcal{U}(t)\mathcal{PL}z + \eta_t +\int^t_0 K(t-s)\mathcal{U}(s)z  ds \, , 
\end{align*}
where the fluctuating forces $\eta$ and memory kernel $K$ are given by
\begin{align}
    \eta_t&:=\mathcal{G}(t)\mathcal{QL}z \, , \\
    K(t)&:=(\eta_t,\mathcal{Q}\mathcal{L}^\dagger z)(z,z)^{-1} \, . \label{equ:2fdt_conclusion}
\end{align}
\end{enumerate}

A semigroup approach using the bounded perturbation theorem is given in sec.~\ref{ssec:semigroup_approach}. (The first statement (i) is given by lemma \ref{lemma:strongly_continuous_semigroup_2} and the second statement (ii) is given by theorem \ref{theorem:gle_2}.) The proof resembles a variation of constants for the approximate orbit maps obtained from the exponential formula for strongly continuous semigroups. A limiting process then yields the desired results without additional assumptions. 

An alternative proof that does not require the bounded perturbation theorem is given in sec.~\ref{ssec:volterra}. Moreover, we neither invoke the variation of constants formula nor the exponential formula for strongly continuous semigroups. Instead, the orbit maps for the orthogonal dynamics are defined by means of solutions to linear Volterra equations. A direct analysis shows that these orbit maps are mild solutions for the abstract Cauchy problem associated to $\overline{\mathcal{QL}}\mathcal{Q}$ (theorem \ref{theorem:mild_solution}).  Using Grönwall's inequality, we then show that these orbit maps constitute a strongly continuous semigroup generated by $\overline{\mathcal{QL}}\mathcal{Q}$ (corollary \ref{corollary:generator}). This proves the first statement (i). In addition, corollary \ref{corollary:generator} yields $\overline{\mathcal{QL}}\mathcal{Q}=\mathcal{QLQ}$. With these results, the second statement (ii) follows by construction due to the uniqueness of solutions to linear Volterra equations (theorem \ref{theorem:gle}). The results from sec.~\ref{sec:non_autonomous_systems} suggest that the same line of argument applies to non-autonomous systems. In order to derive analogous properties of the orthogonal dynamics, however, one would have to treat non-autonomous abstract Cauchy problems, which are generally much more difficult. 

We have allowed for non-unitary time-evolution in a generic setting. Hence, there might be pathological cases where the domain intersection $D(\mathcal{L})\cap D(\mathcal{L}^\dagger)$ is small or even trivial, cf. \cite{arlinskii}. However, under the assumptions of theorem \ref{theorem:time_evolution} for classical statistical mechanics, the space of test functions $C^1_c(\mathbb{R}^n)$ is a subset of the domain intersection. Thus, $D(\mathcal{L})\cap D(\mathcal{L}^\dagger)$ is a dense subspace of observables in $H$. 

For the case of unitary time-evolution, eq.~(\ref{equ:2fdt_conclusion}) is reduced to the well-known second fluctuation dissipation theorem $K(t)=-(\eta_t,\eta_0)(z,z)^{-1}$. Furthermore, the orthogonal dynamics $\{\mathcal{G}(t)\}_{t\in \mathbb{R}}$ is a strongly continuous unitary group. This implies that the fluctuating forces are stationary, i.e., $(\eta_{t+r},\eta_{s+r})=(\eta_t,\eta_s)$ for all $t,s,r \in \mathbb{R}$. 

All in all we have presented a new path towards a rigorous formulation of the Mori projection operator technique under minimal assumptions in a general Hilbert space setting. A second set of proofs validates the main results by means of semigroup theory. 

\section*{Acknowledgment}
We acknowledge funding by the German Research Foundation (DFG) in project 535083866.
We thank Carsten Hartmann, Mario Ayala, Thomas Franosch, Felix H\"ofling, Anja Seegebrecht, Fabian Koch and Klara Resch for useful feedback on our work.

\section*{Data availability statement}
No new data were created or analysed in this study.

\section*{Conflict of interest}
The authors have no conflict of interest.

\appendix

\section{Proof of example \ref{example:givon}}\label{app:example}
\begin{proof}
     Clearly, we have $\mathcal{Q}^2=\mathcal{Q}$ and $\|\mathcal{Q}\|\leq 1$. Since $\mathcal{Q}\neq 0$, we have $\|\mathcal{Q} f\|=\|f\|$ for some $f\in \mathcal{Q}H$. Hence $\|\mathcal{Q}\|\geq 1$. Thus, $\mathcal{Q}$ is a projection of norm one. This implies that $\mathcal{Q}=\mathcal{Q}^\dagger$ is an orthogonal projection \cite[p.~262, Satz (Theorem)~V.5.9]{werner}. 
    
    Recall that the weak derivative $\mathcal{L}$ is skew-adjoint on $L^2(\mathbb{R})$ with domain $D(\mathcal{L})=W^{1,2}(\mathbb{R})$. The space $W^{1,2}(\mathbb{R})$ is the set of weakly differentiable functions with weak derivative in $L^2(\mathbb{R})$, or equivalently, the space of absolutely continuous functions with point-wise derivative in $L^2(\mathbb{R})$. Furthermore, the weak derivative coincides with the point-wise derivative almost everywhere. Next, observe that $\mathcal{Q}f\in D(\mathcal{L})$, if and only if $f$ is absolutely continuous on $\mathbb{R}_+$ with point-wise derivative in $L^2(\mathbb{R}_+)$ and $f(0)=0$. More explicitly, we have
    \begin{align*}
        D(\mathcal{LQ}) &= \{ f\in L^2(\mathbb{R}): \mathcal{Q}f \in D(\mathcal{L})  \} \\
        &= L^2(\mathbb{R}_-) \oplus W^{1,2}_0(\mathbb{R}_+)  \, , \intertext{where}
         W^{1,2}_0(\mathbb{R}_+) &:= \{f\in W^{1,2}(\mathbb{R}_+): f(0)= 0\} \, .
    \end{align*}
    Further, it is straightforward to show that $D(\mathcal{LQ})$ is dense in $L^2(\mathbb{R})$. 
    
    Hence, $\mathcal{QLQ}$ is densely defined and skew-symmetric. This implies that $\mathcal{QLQ}$ is closable, and $\overline{\mathcal{QLQ}}$ is again skew-symmetric \cite[p.~252-255]{reed1980}. The adjoint of $\overline{\mathcal{QLQ}}$ is given by
    \begin{align*}
        \overline{\mathcal{QLQ}}^\dagger &= (\mathcal{QLQ})^\dagger = (\mathcal{LQ})^\dagger \mathcal{Q} =- {(\mathcal{QL})^\dagger}^\dagger \mathcal{Q} = -\overline{\mathcal{QL}}\mathcal{Q} \, ,
    \end{align*}
    where we used \cite[p.~252, Theorem VIII.1]{reed1980} and \cite[p.~330, Theorem~13.2]{rudin1974}. One might expect that $\overline{\mathcal{QLQ}}$ is skew-adjoint and generates a strongly continuous unitary group. However, this is not the case \cite{givon2005}.
    
    We have to show that $\overline{\mathcal{QLQ}} \subsetneq -\overline{\mathcal{QLQ}}^\dagger$. Since $\overline{\mathcal{QLQ}}$ is densely defined, closed, and skew-symmetric, we have \cite[p.~255]{reed1980}
    \begin{align*}
        \overline{\mathcal{QLQ}} \subseteq -\overline{\mathcal{QLQ}}^\dagger  \, .
    \end{align*}
    Hence, it suffices to find a function $f\in D(\overline{\mathcal{QLQ}}^\dagger)$ with $f\notin D(\overline{\mathcal{QLQ}})$. Let $f(x):=\Theta(x)e^{-x}$. For all $g\in D(\overline{\mathcal{QLQ}})$, there exists a sequence $D(\mathcal{QLQ})\ni g_n \to g$ with $\lim_{n\to\infty} \mathcal{QLQ}g_n=\overline{\mathcal{QLQ}}g$. Note that $W^{1,2}_0(\mathbb{R}_+)$ is given by the closure of $C^1_c(\mathbb{R}_+)$ in $W^{1,2}(\mathbb{R}_+)$ \cite[p.~217, Theorem~8.12]{Brezis2011}. Hence, $C^1_c(\mathbb{R}_+)$ is dense in $W^{1,2}_0(\mathbb{R}_+)$. Thus, for all $n\in\mathbb{N}$, there exists a sequence $C^1_c(\mathbb{R}_+) \ni \phi_{mn} \to g_n \in W^{1,2}_0(\mathbb{R}_+)$. Thus, we find 
    \begin{align*}
        (f,\overline{\mathcal{QLQ}}g) &= \lim_{n\to\infty} (f, \mathcal{QLQ}g_n) \\
        &= \lim_{n\to\infty} \int^\infty_0 f(x) \partial_x g_n(x) \, dx \\
            &= \lim_{n\to\infty} \lim_{m\to\infty} \int^\infty_0 f(x) \phi_{mn}'(x) \, dx \\
        &= - \lim_{n\to\infty} \lim_{m\to\infty} \int^\infty_0 f'(x)  \phi_{mn}(x) \, dx   \\
        &= \lim_{n\to\infty} \int^\infty_0 f(x) g_n(x) \, dx \\
        &= (f,g) \, ,
    \end{align*}
    where $'$ denotes the point-wise derivative. This implies
    \begin{align*}
        f &= \overline{\mathcal{QLQ}}^\dagger f \, , \intertext{i.e., if $\mathcal N$ denotes the null space}
        f &\in\mathcal{N}(\overline{\mathcal{QLQ}}^\dagger-1) \, .
    \end{align*}
    According to von Neumann's formula \cite[p.~82, Lemma~2.2]{gallone}, we have
    \begin{align*}
        D(\overline{\mathcal{QLQ}}^\dagger) = D(\overline{\mathcal{QLQ}}) \oplus_G \mathcal{N}(\overline{\mathcal{QLQ}}^\dagger-1)\oplus_G \mathcal{N}(\overline{\mathcal{QLQ}}^\dagger+1) \, ,
    \end{align*}
    where $\oplus_G$ denotes the orthogonal sum induced by the graph norm for $\overline{\mathcal{QLQ}}^\dagger$. This implies $f\in D(\overline{\mathcal{QLQ}}^\dagger)$, but $f\notin D(\overline{\mathcal{QLQ}})$ as desired. 

    Hence, $\overline{\mathcal{QLQ}}\neq -\overline{\mathcal{QLQ}}^\dagger$. Thus, $\overline{\mathcal{QL}}\mathcal{Q} = -\overline{\mathcal{QLQ}}^\dagger$ is not skew-adjoint. According to Stone's theorem \cite[p.~89, Theorem~3.24]{engel}, $\overline{\mathcal{QL}}\mathcal{Q}$ does not generate a strongly continuous unitary group. 
\end{proof}

\section{Proof of theorem \ref{theorem:time_evolution}}\label{app:time_evolution}

For the proof of theorem \ref{theorem:time_evolution}, we require the product rule and partial integration for weakly differentiable functions. Typically, only test functions in $C^\infty_c(\mathbb{R}^n)$ are considered \cite{evans}. Thus, we have to show that the standard results extend to test functions in $C^1_c(\mathbb{R}^n)$. 

\begin{proposition}\label{proposition:weak_derivative}
    Let $\Omega=\mathbb{R}^n$, and $\mu$ be the Lebesgue measure on $\Omega$. Let $f\in W^{1,1}(\Omega)$ and $g \in C^1_c(\Omega)$. Then, $fg\in W^{1,1}(\Omega)$, $\partial_i (fg)= (\partial_i f)g+f(\partial_ig)$, and
    \begin{align*}
        \int_\Omega f (\partial_i g) \, d\mu &= -\int_\Omega (\partial_if)  g \, d\mu \, , 
    \end{align*}
    for all $i=1,\cdots,n$, where $\partial_i$ denotes the weak derivative.  
\end{proposition}
\begin{proof}
    There exists a sequence $C^\infty_c(\Omega)\ni f_n \to f \in W^{1,1}(\Omega)$ \cite[p.~265, Theorem~9.2]{Brezis2011}. From Hölder's inequality  \cite[p.~152, Theorem~7.16]{klenke}, we know that $\phi \to \int_\Omega \phi\psi d\mu $ is a bounded linear functional on $L^1(\Omega)$ for all $\psi \in L^\infty(\Omega)$. Since $\partial_i g, g \in L^\infty(\Omega)$, this implies
    \begin{align*}
        \int_\Omega f(\partial_i g) \, d\mu &= \lim_{n\to\infty} \int_\Omega f_n(\partial_i g)\, d\mu = -\lim_{n\to\infty} \int_\Omega (\partial_i f_n) g \, d\mu= -  \int_\Omega (\partial_i f) g \, d\mu\, .
    \end{align*}
    
    Suppose there exists a sequence $C^\infty_c(\Omega) \ni g_n \to g \in W^{1,\infty}(\Omega)$. For all $\varphi \in C^\infty_c(\Omega)$, we have \cite[p.~246, Theorem~1]{evans}
    \begin{align*}
        \int_\Omega  [(\partial_if) g_n + f(\partial_i g_n)]\varphi  \, d\mu &= - \int_\Omega  fg_n (\partial_i\varphi)  \, d\mu \, .
    \end{align*}
    Due to Hölder's inequality, the sequences $(\partial_if) g_n, f(\partial_i g_n), fg_n$ are Cauchy in $L^1(\Omega)$. Since $\varphi, \partial_i\varphi \in L^\infty(\Omega)$, taking the limit $n\to\infty$ yields
    \begin{align*}
        \int_\Omega  [(\partial_if) g + f(\partial_i g)]\varphi  \, d\mu &= - \int_\Omega  fg (\partial_i\varphi)  \, d\mu \, .
    \end{align*}
    Hence, $\partial_i(fg)= (\partial_if) g + f(\partial_i g)$. Clearly, $fg \in W^{1,1}(\Omega)$. 
    
    It is left to show that there exists a sequence $C^\infty_c(\Omega) \ni g_n \to g \in W^{1,\infty}(\Omega)$. Let $\rho_n$ be a sequence of mollifiers \cite[p.~108]{Brezis2011}, and set $g_n:= \rho_n\star g$, where $\star$ denotes the convolution. Since $g \in C_c(\Omega)$, $g_n\to g$ uniformly \cite[p.~108, Proposition~4.21]{Brezis2011}. Since $g,\rho_n\in C^1_c(\Omega)$, it follows from the Leibniz integral rule \cite[p.~56, Theorem~2.27]{folland} that $\partial_i g_n = \rho_n \star \partial_i g$. Since $\partial_i g \in C_c(\Omega)$ , $\partial_i g_n\to \partial_i g$ uniformly \cite[p.~108, Proposition~4.21]{Brezis2011}. Clearly, $g_n \in C^\infty_c(\Omega)$. Hence, we have $C^\infty_c(\Omega) \ni g_n \to g \in W^{1,\infty}(\Omega)$. This concludes the proof. 
\end{proof}

It is a standard result that the space $C_c(\mathbb{R}^n)$ is dense in $L^2(P)$ for regular Borel measures on $\mathbb{R}^n$ \cite{cohn}. In the following, we convince ourselves that this result applies to the probability measure $P$ from eq.~(\ref{def:probability_measure}). Furthermore, we need to extend this result to the space $C^1_c(\mathbb{R}^n)$. 

\begin{proposition}\label{proposition:dense}
    Let $\Omega = \mathbb{R}^n$ and let $P$ be the probability measure from eq.~(\ref{def:probability_measure}). Then, $C^1_c(\Omega)$ is dense in $L^2(P)$.
\end{proposition}
\begin{proof}
    Suppose that $P$ is regular on $(\Omega,\Sigma)$. Since $\Omega$ is a locally compact Hausdorff space and $\mathcal{B}(\Omega)\subseteq \Sigma$, it follows from \cite[p.~207, Proposition~7.4.3]{cohn} that $C_c(\Omega)$ is dense in $L^2(P)$. Further, the space $C^1_c(\Omega)$ is a subalgebra of $C_0(\Omega)$ which separates points, and for all $\mathbf{x}\in \Omega$ there exists a function in $C^1_c(\Omega)$ that does not vanish at $\mathbf{x}$. Hence, by the Stone-Weierstraß theorem \cite[p.~392, Theorem~D.23]{cohn}, it follows that $C^1_c(\Omega)$ is uniformly dense in $C_c(\Omega)$. Since $P(\Omega)=1$, this implies that $C^1_c(\Omega)$ is dense in $C_c(\Omega)$ w.r.t. the norm on $L^2(P)$. Since $C_c(\Omega)$ is dense in $L^2(P)$, this implies that $C^1_c(\Omega)$ is dense in $L^2(P)$. 

    It is left to show that $P$ is regular. Since $P$ is a finite measure on $(\Omega,\mathcal{B}(\Omega))$, it follows from \cite[p.~34, Proposition~1.5.6]{cohn} that $P$ is regular on $(\Omega,\mathcal{B}(\Omega))$. Since $\Sigma$ is the completion of $\mathcal{B}(\Omega)$ under the Lebesgue measure on $(\Omega,\mathcal{B}(\Omega))$ \cite[p.~32, Proposition~1.5.2]{cohn}, for all $\mathcal{M}\in \Sigma$, there exists $\mathcal{M}',\mathcal{M}'' \in \mathcal{B}(\Omega)$ such that $\mathcal{M}'\subseteq \mathcal{M}\subseteq\mathcal{M}''$ and $\mathcal{M}''-\mathcal{M}'$ has Lebesgue measure zero. By definition, $P(\mathcal{M}''-\mathcal{M}')=0$. This implies $P(\mathcal{M}')=P(\mathcal{M})=P(\mathcal{M}'')$. Since $P$ is regular on $(\Omega,\mathcal{B}(\Omega))$, and $\mathcal{M}',\mathcal{M}'' \in \mathcal{B}(\Omega)$, it follows that $P$ is regular on $(\Omega,\Sigma)$.
\end{proof}

Finally, let us briefly recall some elementary facts from the theory of ordinary differential equations. In addition, we require the flow $\varphi_t$ to be $\Sigma$-$\Sigma$ measurable.

\begin{proposition}\label{proposition:flow}
    Let $\Omega=\mathbb{R}^n$ and $\Sigma$ be the Lebesgue $\sigma$-algebra on $\Omega$. Let $\mathbf{F}:\Omega\to\Omega$ be continuously differentiable with bounded derivative. Then, the flow $\varphi:\mathbb{R}\times\Omega \to \Omega$ generated by the vector field $\mathbf{F}$ is Lipschitz continuous in $\Omega$ (uniformly on compact intervals) and $\Sigma$-$\Sigma$ measurable (for all $t\in\mathbb{R}$).
\end{proposition}
\begin{proof}
    First, we show that $\mathbf{F}$ is Lipschitz continuous. Let $\boldsymbol{\gamma}(t):=t\mathbf{x}+(1-t)\mathbf{y}$ for $\mathbf{x},\mathbf{y}\in\Omega$ and $t\in [0,1]$. Then, by the mean value theorem \cite[p.~113, Theorem~15.19]{rudin1964}, we have for some $t_0\in [0,1]$
    \begin{align*}
        |\mathbf{F}(\mathbf{x})-\mathbf{F}(\mathbf{y})| &\leq \Big\lvert \frac{d}{dt} \mathbf{F}(\boldsymbol{\gamma}(t)) \Big\rvert_{t=t_0} = \left| \mathbf{F}'(\boldsymbol{\gamma}(t_0))\cdot (\mathbf{x}-\mathbf{y}) \right|  \leq M |\mathbf{x}-\mathbf{y}| \, ,
    \end{align*}
    where $M:=\sup_\Omega \|\mathbf{F}'\|$. Hence, $\mathbf{F}$ is Lipschitz continuous. By the Picard-Lindelöf theorem \cite[p.~214, Theorem~11]{poeschel2}, it follows that there exists a unique solution $\varphi(\mathbf{x}):\mathbb{R}\to\Omega$ of the initial value problem $\frac{d}{dt} \varphi_t(\mathbf{x}) = \mathbf{F}(\varphi_t(\mathbf{x}))$, $\varphi_0(\mathbf{x})=\mathbf{x}$, for all initial values $\mathbf{x}\in \Omega$. Further, it follows from Grönwall's inequality that $\varphi_t$ is Lipschitz continuous (uniformly on compact intervals) \cite[p.~47, Theorem~4.2]{gruene}. Moreover, we have $\varphi_0=\text{id}_\Omega$ and $\varphi_{t+s} = \varphi_t \circ \varphi_s$ \cite[p.~42-43]{gruene}. 

    It is left to show that $\varphi_t$ is $\Sigma$-$\Sigma$ measurable. Since $\varphi_t$ is continuous, $\varphi_t$ is $\mathcal{B}(\Omega)$-$\mathcal{B}(\Omega)$ measurable \cite[p.~189, Lemma~7.2.1]{cohn}. Since $\varphi_t^{-1}=\varphi_{-t}$ is Lipschitz continuous, $\varphi_t^{-1}$ maps sets of Lebesgue measure zero onto sets of Lebesgue measure zero \cite[p.~153, Lemma~7.25]{rudin1987}. Since $\Sigma$ is the completion of $\mathcal{B}(\Omega)$ under the Lebesgue measure on $(\Omega,\mathcal{B}(\Omega))$ \cite[p.~32, Proposition~1.5.2]{cohn}, this implies that $\varphi_t$ is $\Sigma$-$\Sigma$ measurable for all $t\in\mathbb{R}$. 
\end{proof}

\begin{proof}[Proof of theorem \ref{theorem:time_evolution}]
Let $X=(C_0(\Omega),\sup_\Omega|\cdot|)$. For $x,y\in L^2(P)$, we write $(x,y) := (x,y)_{L^2(P)}$. First, we show that the composition operator $(\mathcal{K}(t), C^1_c(\Omega))$ is a bounded linear operator on $L^2(P)$, where $\mathcal{K}(t)$ is defined by eq.~(\ref{def:koopman}). Let $x,y\in C^1_c(\Omega)$ arbitrary. Since convergence in $X$ implies convergence in $L^2(P)$, we have
\begin{align}
    \frac{d}{dt} ( \mathcal{K}(t)x, \mathcal{K}(t)y) &= (\mathbf{F}\cdot\boldsymbol{\nabla}  \mathcal{K}(t)x, \mathcal{K}(t)y)+(  \mathcal{K}(t)x,\mathbf{F}\cdot\boldsymbol{\nabla} \mathcal{K}(t)y) \, , \label{equ:proof_stat_mech_1}
\end{align}
where we used the fact that $\{\mathcal{K}(t)\}_{t\in\mathbb{R}}$ is a strongly continuous group on $X$, and $(\mathbf{F}\cdot\boldsymbol{\nabla},C^1_c(\Omega))$ is a core of the generator \cite[p.~91-92]{engel}, cf. sec.~\ref{sec:classical_statistical_mechanics}. 
Since $C^1_c(\Omega)$ is a $\{\mathcal{K}(t)\}_{t\in\mathbb{R}}$-invariant subspace, we have $ \mathcal{K}(t)x \in C^1_c(\Omega)$ and $\mathcal{K}(t)y \in C^1_c(\Omega)$. Since $\rho\mathbf{F} \in W^{1,1}(\Omega)$, it follows from proposition \ref{proposition:weak_derivative} that $\rho\mathbf{F}\mathcal{K}(t)x \in W^{1,1}(\Omega)$, and
\begin{align*}
    (  \mathcal{K}(t)x,\mathbf{F}\cdot\boldsymbol{\nabla} \mathcal{K}(t)y) &\overset{(\ref{def:scalar_product_L2P})}{=} \int_\Omega   [\mathcal{K}(t)x] \rho \mathbf{F}\cdot \boldsymbol{\nabla}  \mathcal{K}(t)y \, d\mu  \\
    &= - \int_\Omega \text{div}(  \rho \mathbf{F}\mathcal{K}(t)x)   \mathcal{K}(t)y \, d\mu  \\
    &= -(\mathbf{F}\cdot \boldsymbol{\nabla}  \mathcal{K}(t)x, \mathcal{K}(t)y) - \left(\frac{1}{\rho} \text{div}(\rho\mathbf{F}) \mathcal{K}(t)x,  \mathcal{K}(t)y\right) \, .
\end{align*}
Inserting into eq.~(\ref{equ:proof_stat_mech_1}) yields
\begin{align}
        \frac{d}{dt} ( \mathcal{K}(t)x, \mathcal{K}(t)y) &= - \left(\frac{1}{\rho} \text{div}(\rho\mathbf{F}) \mathcal{K}(t)x,  \mathcal{K}(t)y \right)  \, . \label{equ:proof_stat_mech_2}
\end{align}
According to Hölder's inequality \cite[p.~152, Theorem~7.16]{klenke}, we have
\begin{align*}
    \frac{d}{dt} ( \mathcal{K}(\pm t)x, \mathcal{K}(\pm t)x) &\leq \left\|\frac{1}{\rho} \text{div}(\rho\mathbf{F}) [\mathcal{K}(\pm t)x]^2\right\|_{L^1(P)}\\
    &\leq \left\|\frac{1}{\rho} \text{div}(\rho\mathbf{F})\right\|_{L^\infty(P)}\left\|[\mathcal{K}(\pm t)x]^2\right\|_{L^1(P)}\\
    &= 2\omega_0  ( \mathcal{K}(\pm t)x, \mathcal{K}(\pm t)x) \, ,
\end{align*}
for all $t\geq 0$, where $\omega_0 := \frac{1}{2}\left\|\rho^{-1}\text{div}(\rho\mathbf{F})\right\|_{L^\infty(P)}$, cf. the discussion after theorem \ref{theorem:time_evolution} and \cite[Sec.~III.A]{zhu18}\cite{davies}. 
Due to Grönwall's inequality, we conclude 
\begin{align}
    \| \mathcal{K}(t)x\|^2_{L^2(P)} &\leq  e^{2\omega_0 |t|} \|x\|^2_{L^2(P)} \, , \label{equ:growth_bound_0}
\end{align}
for all $x\in C^1_c(\Omega)$ and $t\in \mathbb{R}$. 

Since $C^1_c(\Omega)$ is dense in $L^2(P)$, the composition operator $(\mathcal{K}(t), C^1_c(\Omega))$ is a densely defined, bounded linear operator on $L^2(P)$. Hence, we may define the time evolution operator $\mathcal{U}(t)$ on $L^2(P)$ by the unique continuous extension of the composition operator \cite[p.~52, Satz (Theorem)~II.1.5]{werner}, such that ${\mathcal U}(t)$ is a linear bounded operator from $L^2(P)$ to itself, 
\begin{align*}
    \mathcal{U}(t)x &= \mathcal{K}(t)x \, , \quad x\in C^1_c(\Omega) \, .
\end{align*}
Since $C^1_c(\Omega)$ is dense in $L^2(P)$, we obtain the growth bound
\begin{align*}
    \|\mathcal{U}(t) \| & \overset{(\ref{equ:growth_bound_0})}{\leq} e^{\omega_0 |t|} \, .
\end{align*}
For $C^1_c(\Omega) \ni x_n\to x \in L^2(P)$, we have
\begin{align*}
    \mathcal{U}(t+s)x &= \lim_n\mathcal{K}(t+s)x_n = \lim_n\mathcal{K}(t)\mathcal{K}(s)x_n = \lim_n\mathcal{U}(t)\mathcal{U}(s)x_n = \mathcal{U}(t)\mathcal{U}(s)x \, .
\end{align*}
Since $C^1_c(\Omega)$ is dense in $L^2(P)$, $\{\mathcal{U}(t)\}_{t\in\mathbb{R}}$ is a group. Furthermore, for all $x\in C^1_c(\Omega)$, the orbit map $t\mapsto \mathcal{K}(t)x \in X$ is continuous. Thus, the orbit map $t\mapsto \mathcal{U}(t)x \in L^2(P)$ is continuous for all $x\in C^1_c(\Omega)$. Since $C^1_c(\Omega)$ is dense in $L^2(P)$, it follows from \cite[p.~38, Proposition~5.3c]{engel} that $\{\mathcal{U}(t)\}_{t\in\mathbb{R}}$ is a strongly continuous group. Let us denote its generator by $(\mathcal{L},D(\mathcal{L}))$. Since $C^1_c(\Omega)$ is dense in $L^2(P)$ and $\{\mathcal{U}(t)\}_{t\in\mathbb{R}}$-invariant, $(\mathbf{F}\cdot\boldsymbol{\nabla},C^1_c(\Omega))$ is a core of $\mathcal{L}$ \cite[p.~53, Proposition~1.7]{engel}. Since $\mathcal{L}$ is closed, we conclude
\begin{align*}
    \mathcal{L} = \overline{(\mathbf{F}\cdot\boldsymbol{\nabla},C^1_c(\Omega))}^{L^2(P)} \, .
\end{align*}

Now, let $\text{div}(\rho\mathbf{F}) =0$ $\mu$-almost everywhere. Then, we have for all $x,y \in C^1_c(\Omega)$
\begin{align*}
    \frac{d}{dt} ( \mathcal{K}(t)x, \mathcal{K}(t)y) &\overset{(\ref{equ:proof_stat_mech_2})}{=} 0 \, .     
\end{align*}
This implies for all $x,y \in C^1_c(\Omega)$
\begin{align}
    ( \mathcal{K}(t)x, \mathcal{K}(t)y) &= (x,y) \, . \label{equ:proof_stat_mech_3}
\end{align}
Since $C^1_c(\Omega)$ is dense in $L^2(P)$, $\mathcal{U}(t)$ is a linear isometry on $L^2(P)$ for all $t\in\mathbb{R}$. 

We have to show that $P=P\circ\varphi^{-1}_t$, where $P\circ\varphi^{-1}_t$ denotes the image measure under the map $\varphi_t$ \cite[p.~75-76]{cohn}, and $\varphi:\mathbb{R}\times \Omega\to\Omega$ is the flow introduced in sec~\ref{sec:classical_statistical_mechanics}. By proposition \ref{proposition:flow}, $\varphi_t$ is $\Sigma$-$\Sigma$ measurable. In particular, the composition of a Lebesgue measurable function with $\varphi_t$ is again Lebesgue measurable. 
Let $f\in C_c(\Omega)$ arbitrary. Since $P$ is a finite measure, $f\circ\varphi_t$ is $P$-integrable. According to \cite[p.~(76), Proposition 2.6.8]{cohn}, $f$ is $P\circ\varphi^{-1}_t$-integrable and
\begin{align}
    \int f \circ \varphi_t \, dP &= \int f \, d(P\circ\varphi^{-1}_t) \, . \label{equ:change_of_variables}
\end{align}
Let $g\in C_c(\Omega)$. 
By the Stone-Weierstraß theorem \cite[p.~392, Theorem~D.23]{cohn}, $C^1_c(\Omega)$ is uniformly dense in $C_c(\Omega)$. Thus, there exist sequences $C^1_c(\Omega) \ni f_n \to f$ uniformly and $C^1_c(\Omega) \ni g_n \to  g$ uniformly.
Since $P$ and $P\circ\varphi^{-1}_t$ are finite measures, uniform convergence implies convergence in $L^2(P)$ and $L^2(P\circ\varphi^{-1}_t)$. Thus, we have
\begin{align*}
    (f,g) &= \lim_{mn} (f_m,g_n) \\
    &\overset{(\ref{equ:proof_stat_mech_3})}{=}\lim_{mn}(f_m\circ\varphi_t,g_n\circ\varphi_t) \\
    &\overset{(\ref{equ:change_of_variables})}{=} \lim_{mn}(f_m,g_n)_{L^2(P\circ\varphi^{-1}_t)} \\
    &= (f,g)_{L^2(P\circ\varphi^{-1}_t)} \, .
\end{align*}
For $g\in C_c(\Omega)$ with $g\rvert_{\text{supp}(f)}=1$, this implies
\begin{align}
    \int f dP &= \int f d(P\circ\varphi^{-1}_t) \, .
\end{align}
Since $f\in C_c(\Omega)$ is arbitrary and $f\to \int fdP$ is a positive linear functional on $C_c(\Omega)$, it follows from \cite[p.~40, Theorem 2.14]{rudin1987} that there exists a $\sigma$-algebra $\mathcal{A}\supset \mathcal{B}(\Omega)$ and a \textit{unique} positive measure $\nu$ on $\mathcal{A}$ such that $\int f d\nu = \int f dP$ for all $f\in C_c(\Omega)$. In addition $\nu$ is complete on $\mathcal{A}$, thus $\Sigma \subset \mathcal{A}$. In particular, $\nu$ is the only positive measure on $\Sigma$ such that $\int f d\nu = \int f dP$ for all $f\in C_c(\Omega)$. Hence, we conclude $P(\mathcal{M})=P\circ\varphi^{-1}_t(\mathcal{M})$ for all $\mathcal{M}\in\Sigma$, i.e.,
\begin{align}
    P &= P\circ \varphi^{-1}_t  \, . \label{equ:invariant_measure}
\end{align}
This implies for all $x\in L^2(P)$
\begin{align*}
    \|x\circ\varphi_t\|_{L^2(P)} &\overset{(\ref{equ:invariant_measure})}{=} \|x\circ\varphi_t\|_{L^2(P\circ\varphi^{-1}_{-t})} = \|x\circ\varphi_t\circ\varphi_{-t}\|_{L^2(P)} = \|x\|_{L^2(P)} \, ,
\end{align*}
where we used \cite[p.~76, Proposition~2.6.8]{cohn}. Hence, the composition $x\mapsto x\circ\varphi_t$ is a bounded linear operator on $L^2(P)$. Due to the uniqueness of a continuous extension \cite[p.~52, Satz (Theorem)~II.1.5]{werner}, we conclude that $\mathcal{U}(t)x=x\circ\varphi_t$ for all $x\in L^2(P)$ and $t\in\mathbb{R}$. 
\end{proof}

\section{Addendum}\label{sec:addendum}

Theorem \ref{theorem:gle_2} gives a proof of the GLE and 2FDT for the Mori projection.
We note that the same result is obtained by a variation of constants,
which greatly simplifies the semigroup approach presented in section \ref{ssec:semigroup_approach}. We therefore add a short proof of the GLE by means of the variation of constants formula for strongly continuous semigroups. 

\begin{theorem*}
Let $\mathcal{P}$ and $\mathcal{Q}=1-\mathcal{P}$ be orthogonal projections on a complex Hilbert space $H$. Let $\mathcal{L}$ be the generator of a strongly continuous semigroup $\mathcal{U}(t)$ on $H$. Let $z\in D(\mathcal{L})$ and let $\mathcal{PL}$ be bounded. Then,
\begin{align*}
\frac{d}{dt}\mathcal{U}(t)z &= \mathcal{U}(t)\mathcal{P}\mathcal{L}z+\mathcal{G}(t)\mathcal{Q}\mathcal{L}z + \int^t_0 \mathcal{U}(t-s)\widehat{\mathcal{PL}}\mathcal{G}(s)\mathcal{Q}\mathcal{L}z \, ds \, , 
\end{align*}
where $\mathcal{G}(t)$ is the strongly continuous semigroup generated by $\mathcal{QL}$ and $\widehat{\mathcal{PL}}:H\to H$ is the continuous extension of $\mathcal{PL}$. 
\end{theorem*}
\begin{proof}
    Since $\mathcal{PL}$ is densely defined \cite[p.~51, Theorem~1.4]{engel} and bounded, its unique continuous extension $\widehat{\mathcal{PL}}:H\to H$ is well-defined. By the bounded perturbation theorem \cite[p.~158, Theorem~1.3]{engel}, the operator $\mathcal{QL}=\mathcal{L}-\widehat{\mathcal{PL}}$ generates a strongly continuous semigroup $\mathcal{G}(t)$. In addition, the variation of constants formula holds for all $x\in H$ \cite[p.~161, Corollary~1.7]{engel}:
    \begin{align*}
        \mathcal{G}(t)x &= \mathcal{U}(t)x-\int^t_0 \mathcal{U}(t-s)\widehat{\mathcal{PL}}\mathcal{G}(s)x ds \, .
    \end{align*}
    For $x= \mathcal{QL}z$, we obtain the generalized Langevin equation
    \begin{align*}
        \frac{d}{dt}\mathcal{U}(t)z &= \mathcal{U}(t)\mathcal{PL}z+\mathcal{U}(t)\mathcal{QL}z \\
        &= \mathcal{U}(t)\mathcal{P}\mathcal{L}z+\mathcal{G}(t)\mathcal{Q}\mathcal{L}z + \int^t_0 \mathcal{U}(t-s)\widehat{\mathcal{PL}}\mathcal{G}(s)\mathcal{Q}\mathcal{L}z \, ds \, .
    \end{align*}
\end{proof}

The Mori projection is defined by $\mathcal{P}x:= (x,z)(z,z)^{-1}z$, where $(.,.)$ denotes the scalar product with complex conjugation in its second argument. For $0\neq z \in D(\mathcal{L}^\dagger)$, the operator $\mathcal{PL}$ is bounded with continuous extension
\begin{align*}
    \widehat{\mathcal{PL}}x &= (x,\mathcal{L}^\dagger z)(z,z)^{-1} z \, .
\end{align*}
For $0\neq z \in D(\mathcal{L}^\dagger)\cap D(\mathcal{L})$, we obtain the GLE for the Mori projection:
\begin{align*}
    \frac{d}{dt}\mathcal{U}(t)z &= \mathcal{U}(t)\mathcal{P}\mathcal{L}z + \mathcal{G}(t)\mathcal{Q}\mathcal{L}z+\int^t_0   K(t-s)\mathcal{U}(s)z \, ds \, , \\
     K(t)&:= (\mathcal{G}(t)\mathcal{Q}\mathcal{L}z,\mathcal{L}^\dagger z)(z,z)^{-1} \, .
\end{align*}

We have to show that this result is equivalent to Theorem 2.5. Since $\mathcal{G}(t)$ is a strongly continuous semigroup with generator $\mathcal{QL}$, its orbit maps are the unique mild solutions of the associated abstract Cauchy problem \cite[p.~146, Proposition 6.4]{engel}. In other words, the orbit map $t\to \mathcal{G}(t)x$ is the unique continuous function that satisfies the equation 
\begin{align*}
    \mathcal{G}(t)x-x &= \mathcal{QL}\int^t_0 \mathcal{G}(s)x \, ds \, . 
\end{align*}
This implies that $\mathcal{G}(t)x \in \mathcal{Q}H$ for all $x\in\mathcal{Q}H$ and $t\geq 0$. Hence, the memory kernel takes the form $K(t)= (\mathcal{G}(t)\mathcal{Q}\mathcal{L}z,\mathcal{Q}\mathcal{L}^\dagger z)(z,z)^{-1}$, in agreement with Theorem \ref{theorem:gle_2}. 

Finally, we note that the orbit maps $t\to\mathcal{G}(t)x$ coincide with the corresponding orbit maps of Theorem \ref{theorem:gle_2} for all $x\in\mathcal{Q}H$. In Theorem \ref{theorem:gle_2}, the orbit maps of the orthogonal dynamics are given by $t\to \tilde{\mathcal{G}}(t)x$, where $\tilde{\mathcal{G}}(t)$ is the strongly continuous semigroup generated by $\overline{\mathcal{QL}} \mathcal{Q}$. The operator $\mathcal{QL}$ with domain $D(\mathcal{L})$ is closed, since it is a generator \cite[p.~51, Theorem~1.4]{engel}. Hence, the closure is given by $\overline{\mathcal{QL}}=\mathcal{QL}$. The subspace $\mathcal{Q}H$ is a closed invariant subspace for both semigroups. Therefore, both semigroups possess a subspace semigroup on $\mathcal{Q}H$ \cite[p.~43]{engel}. Both of their generators are given by the part of $\mathcal{QL}$ in $\mathcal{Q}H$ \cite[p.~60-61]{engel}. This implies $\mathcal{G}(t)x=\tilde{\mathcal{G}}(t)x$ for all $x\in\mathcal{Q}H$. 

\bibliographystyle{plain}
\bibliography{MAIN}

\end{document}